\documentclass[preprint,citeautoscript,superscriptaddress,amsmath,amssymb,10pt,twocolumn]{revtex4-2}

\newcommand{\ket}[1]{\left|#1\right\rangle}
\newcommand{\bra}[1]{\left\langle#1\right|}

\usepackage{graphicx}
\usepackage{color}
\usepackage{upgreek}
\usepackage{float}
\usepackage{lipsum}
\usepackage{mathrsfs}
\usepackage{booktabs}
\usepackage{url}
\usepackage{xcolor}
\usepackage[colorlinks,citecolor=teal,linkcolor=teal,urlcolor=teal]{hyperref}
\usepackage{times}

\usepackage{wasysym} % for \diameter
\usepackage{enumerate}
\usepackage{stmaryrd} % for QECC brackets
\usepackage{ragged2e} % for justifying in figure captions
\usepackage{bbm}

\newcommand{\FF}{\mathbb{F}}

\newcommand{\RR}{\mathbb{R}}
\newcommand{\CC}{\mathbb{C}}

\newcommand{\Tr}{\text{Tr}}

\newcommand{\wt}{\mathrm{wt}}
\newcommand{\imag}{\texttt{i}}

\newcommand{\add}{}

\usepackage{amsthm}

\newtheorem{theorem}{Theorem}
\newtheorem{corollary}[theorem]{Corollary}

\newcommand{\fu}{Dahlem Center for Complex Quantum Systems, Freie Universit{\"a}t Berlin, 14195 Berlin, Germany}
\newcommand{\fzj}{Institute for Theoretical Nanoelectronics (PGI-2), Forschungszentrum J\"ulich, 52428 J\"ulich, Germany}
\newcommand{\HZB}{Helmholtz-Zentrum Berlin für Materialien und Energie, Hahn-Meitner-Platz 1, 14109 Berlin, Germany}

\begin{document}

\author{Daniel Miller}
\affiliation{\fu}
\affiliation{\fzj}
 
\author{Jens Eisert}
\affiliation{\fu}
\affiliation{\HZB}

\title{Detecting entanglement from few partial transpose moments\\ and their decay via weight enumerators}

\date{\today}

\begin{abstract}
In light of the rapid experimental development of quantum devices, tools for estimating quantum properties, particularly entanglement, must be improved to keep pace.
The $p_3$-PPT criterion is an experimentally viable relaxation of the well-known positive partial transposition (PPT) criterion for the certification of quantum entanglement.
Recently, it has been generalized to various families of entanglement criteria based on the PT moments $p_k=\Tr[(\rho^\Gamma)^k]$,
where $\rho^\Gamma$ denotes the partially transposed density matrix of a quantum state $\rho$.
While most of these generalizations are strictly more powerful than the $p_3$-PPT criterion, their $m$-th level versions usually rely on the availability of $p_k$ for all moment orders $k\le m$.
Here, we show that one can alternatively compare any three PT moments of   orders $k < l < m$, 
which can significantly reduce experimental overheads.
More precisely, we show that any state satisfying 
$p_{l} > p_k ^{x} p_m^{1-x} $ must be entangled, where $x=(m-l)/(m-k)$.
Using the example of locally depolarized Greenberger-Horne-Zeilinger (GHZ) states, 
a family of states that has been used as a test
for the quantum character of large-scale states prepared on quantum
computers, we identify the most promising versions of these three-moment criteria and compare their performance with a broad range of  entanglement criteria.
In the case of globally depolarized stabilizer states, we prove that having access to $p_k$ for $k \le 5$ is sufficient to reproduce the full PPT criterion.
More generally, we show that the Stieltjes-$m$ criterion is as powerful as the PPT criterion whenever $\rho^\Gamma$ has no more than $(m+1)/2$ distinct eigenvalues.
Finally, we introduce a notion of quantum weight enumerators that capture the decay of $p_k$ under local white noise for arbitrary quantum states and illustrate this concept
for an absolutely maximally entangled (AME) state.
Our results contribute to the growing body of literature on higher-moment PPT relaxations and modern applications of weight enumerators in quantum error correction and information theory.
\end{abstract}

\maketitle

\section{Introduction}

Once largely theoretical, ideas from quantum information science now underpin large-scale experimental quantum computing platforms~~\cite{arute_quantum_supremacy_2019, kim_evidence_for_2023, bluvstein_logical_quantum_2024, acharya_quantum_error_2025, MindTheGaps}. 
Quantum computing is no longer a distant vision but is now closely linked to rapidly advancing experimental systems.
As a result, the demands for certification, benchmarking, and the detection of quantum properties~\cite{eisert_quantum_certification_2020, PRXQuantum.2.010201} have become substantially more stringent. 
Greater emphasis is now placed on the robustness of methods and their experimental feasibility under realistic data conditions, since data acquisition can be involved, costly, and complex.
In this context, the present work is situated.
We improve a known family of entanglement detection methods, enhance their efficiency, and connect them to powerful concepts from quantum error correction~\cite{Roads}.

%Entanglement is the fundamental \dm{fundamental 3 times in 3 sentences...} structural element of quantum information science.
Entanglement is a ubiquitous resource in quantum technology.
It underlies every advantage that quantum systems have over classical ones in information-processing tasks and is what fundamentally distinguishes quantum theory from classical statistical theory.
For this reason, researchers in the early days of the field developed methods to detect entanglement from experimental data.
The first such tests were based on violations of Bell inequalities~\cite{freedman_experimental_test_1972,  aspect_experimental_realization_1982, weihs_violation_of_1998,
hensen_loophole_free_2015, storz_loophole_free_2023}.
While Bell inequality violations are fundamental in ruling out the possibility that nature is governed by a hidden-variable model,
they only enable the detection of so-called {Bell-nonlocal} states.
However, there exist states that are entangled, yet do not exhibit Bell nonlocality. 
A well-known example is the two-qubit Werner state
\begin{align}
    \rho_\text{Werner}(\varepsilon) = (1-\varepsilon)\ket{\Psi^-} \!\bra{\Psi^- } + \varepsilon \frac{\mathbbm1}{4} \, ,
\end{align}
where $\ket{\Psi^-} =  (\ket{0,1}-\ket{1,0}) / \sqrt{2}$ denotes the singlet~\cite{werner_quantum_states_1989}.
The Werner state is entangled for all $\varepsilon< 2/3$,
but only for $\varepsilon < 1 - 1/\sqrt{2}$ it violates a Bell inequality~\cite{horodecki_violating_bell_1995}.
% That said, using violations of Bell inequalities is often not very practical.
A concept that is much more amenable to realistic experimental data is that of a linear entanglement witness \cite{terhal_detecting_quantum_2002, guehne_detection_of_2002, guehne_entanglement_detection_2009,Hangleiter}.
It is able to directly detect entanglement, including 
the entanglement in quantum states that do not exhibit Bell nonlocality.
In particular, if the fidelity with respect to an entangled stabilizer state exceeds 50\,\%, this certifies the presence of entanglement. 
Fidelity experiments have been successfully implemented across various platforms~\cite{bourennane_experimental_detection_2004, leibfried_creation_of_2005, 
huang_experimental_generation_2011, zhang_experimental_greenberger_2015, wang_18_qubit_2018, omran_generation_of_2019, mooney_generation_and_2021, thomas_efficient_generation_2022, cao_generation_of_2023, moses_a_race_2023, kam_characterization_of_2024, javadiabhari_big_cats_2025}.
A limitation of fidelity-based approaches, however, is that they rely on prior knowledge of the target state. 
If such information is not applicable, 
there exist multiple alternatives.
For example, if the purity, $\Tr[\rho^2]$, of a quantum state $\rho$ exceeds that of any of its marginals, $\rho_A =\Tr_B[\rho]$, then $\rho$ must be entangled across the bipartition $A\vert B$.
While there are fundamental scalability bottlenecks for learning $\Tr[\rho^2] $ in the single-copy setting~\cite{chen_exponential_separations_2022},
it can be efficiently estimated by preparing two copies of $\rho$ on the same quantum device,
followed by pairwise Bell measurements on corresponding qubits.
Also such purity-based entanglement detection has been experimentally implemented across various platforms~\cite{bovino_direct_measurement_2005, islam_measuring_entanglement_2015, kaufmann_quantum_thermalization_2016, bluvstein_a_quantum_2022, miller_experimental_measurement_2024}.
By applying a different postprocessing procedure to the raw data obtained in a two-copy purity measurement, one can also efficiently lower bound the state’s concurrence, an important multipartite entanglement measure~\cite{carvalho_decoherence_and_2004, mintert_concurrence_of_2005, mintert_observable_entanglement_2007, aolita_scalable_method_2008}.

A more practical test for entanglement---the one upon which the present work is built---is provided by the
powerful
%However, many states cannot be detected by either of the above criteria, yet are nevertheless entangled.
%An extremely powerful concept is the 
\emph{positive partial transposition} (PPT) criterion~\cite{peres_separability_criterion_1996,
horodecki_separability_1996, rains_a_semidefinite_2001,
plenio_logarithmic_negativity_2005}.
The partial transposition of a state $\rho$ acting on a bipartite Hilbert space $\mathcal{H}_A\otimes \mathcal{H}_B$ is denoted by $\rho^{\Gamma}$ and defined via
\begin{align}  \label{eq:partial_transpose_definition}
\bra{a,b}\rho^{\Gamma}\ket{a',b'}
=
\bra{a',b} \rho\ket{a,b'} \, ,
\end{align}
where ${\ket{a}},{\ket{a'}}$ are basis vectors of $\mathcal{H}_A$ and
${\ket{b}},{\ket{b'}}$ are basis vectors of $\mathcal{H}_B$.
While $\rho^\Gamma$ depends on the choice of basis,
we are usually only interested in the basis-independent spectrum of eigenvalues of the operator $\rho^\Gamma$. 
The PPT criterion asserts that the existence of a negative eigenvalue of $\rho^\Gamma$ implies that $\rho$ is entangled across the bipartition $A\vert B$.
If this is the case, we say that $\rho$ is \emph{negative partial transpose} (NPT) entangled across $A\vert B$.
However, computing the full spectrum of $\rho^\Gamma$ requires solving an eigenvalue problem whose dimension grows exponentially with the number of qubits, and since the full density matrix 
$\rho$ is typically not accessible in experiments, constructing and diagonalizing 
$\rho^\Gamma$ becomes impractical beyond a few qubits.

To nevertheless somewhat exploit the power of the PPT criterion,   recent approaches have introduced more practical entanglement detection methods based on the so-called PT moments,
\begin{align}
    p_k = \Tr\!\left[ (\rho^\Gamma)^k \right] \, ,
\end{align}
of the partially transposed state, where $k\in \mathbb{N}$~\cite{zhou_single_copies_2020, elben_mixed_state_2020, yu_optimal_entanglement_2021, neven_symmetry_resolved_2021, liu_detecting_entanglement_2022, ali_partial_transpose_2023, carrasco_entanglement_phase_2024, tarabunga_quantifying_mixed_2025, bradshaw_a_closed_2025}.
These moments can be efficiently estimated by measuring the multicopy state $\rho^{\otimes k}$ using the $k$-cycle test~\cite{horodecki_method_for_2002, ekert_direct_estimation_2002, leifer_measuring_polynomial_2004, 
carteret_noiseless_quantum_2005,
subasi_entanglement_spectroscopy_2019, oszmaniec_measuring_relational_2024, quek_multivariate_trace_2024}.
Such experiments are difficult for large $k$ and, to the best of our knowledge,
have only been performed up to $k=3$ \cite{PhysRevLett.134.210201}.

In this context, two relaxations of the PPT criterion have been put forward:
the Descartes and  the Stieltjes criteria.
Let $\lambda_1,\ldots,\lambda_{2^n}$ denote the eigenvalues of $\rho^\Gamma$.
Moreover, denote by  $e_m(x_1,\ldots, x_{2^n})$ the elementary symmetric polynomial of degree $m$ in $2^n$ variables.
The Descartes-$m$ criterion exploits that $e_m(\lambda_1,\ldots, \lambda_{2^n})<0$ is only possible if at least one of the eigenvalues is negative, a sufficient condition for $\rho$ to be NPT entangled across $A\vert B$.
Since the real number $e_m(\lambda_1,\ldots, \lambda_{2^n})$ only depends on the PT moments of order $k\le m$~\cite{neven_symmetry_resolved_2021, bradshaw_a_closed_2025},
one can test the Descartes-$m$ criterion as long as $p_1,\ldots, p_m$ are available.
Also the Stieltjes-$m$ criterion makes use of the first $m$ PT moments.
Here, however, one defines the Hankel matrix $H_m \in \RR^{ (m+1)/2 \times (m+1)/2}$, assuming $m$ to be odd, 
whose entries are given by
\begin{align} \label{eq:stieltjes_definition}
    [H_m]_{i,j} = p_{i+j-1} \,,
\end{align}
where $i,j \in\{1,\ldots, (m+1)/2\}$.
If $\rho$ is PPT, then $H_m$ must be a \emph{positive semidefinite} (PSD) matrix~\cite{neven_symmetry_resolved_2021, yu_optimal_entanglement_2021}.
For $m=3$, we have 
\begin{align}
    H_3 = \begin{bmatrix}
        1 & p_2\\ p_2 &p_3 
    \end{bmatrix} \, 
\end{align}
which is PSD if and only if (iff) $p_2^2\le p_3$.
Conversely, every state with $p_2^2>p_3$ is NPT entangled. 
This is the celebrated $p_3$-PPT criterion~\cite{elben_mixed_state_2020}.

For completeness, let us also briefly review the optimal $p_m$-PPT  ($p_m$-OPPT) criterion~\cite{yu_optimal_entanglement_2021}.
In this approach, given $p_1,\ldots,p_{m-1}\in [0,1]$, one computes the largest possible value of $p_m$ for which the moments could stem from a PPT state.
Unfortunately, this approach is computationally prohibitive, as it entails solving a non-convex optimization problem in $2^n$ variables.
In Ref.~\cite{neven_symmetry_resolved_2021}, a closely related optimization problem is treated using Lagrange multipliers.
To our knowledge, these optimization problems have been solved analytically only for the case $m\le4$,
where the $p_m$-OPPT constitutes the best available PPT relaxation.
As we will see in Secs.~\ref{sec:stieltjes_equiv_to_ppt} and~\ref{sec:ghz_thresholds},
however, increasing the moment order can be highly beneficial.

On the other hand, estimating all moments $p_1,\ldots, p_m$ becomes an experimentally daunting challenge when the moment order $m$ grows large.
To ease this burden, we point out that the Stieltjes-$m$ criterion also implies that every $j\times j $ principal minor of $H_m$ must be PSD for PPT states.
In particular, considering   $2\times2$  principal minors of $H_m$ shows that
\begin{align} \label{eq:hankel_minor_criterion}
    p_{(k+m)/2}^2 > p_k p_m
\end{align}
implies that $\rho$ is NPT entangled.
Anticipating our subsequent generalization in Thm.~\ref{thm:ppt_relaxation_hoelder},
we refer to Eq.~\eqref{eq:hankel_minor_criterion}
as the $(k,\, (k+m)/2,\, m)$-PPT criterion.
Applying this criterion requires knowledge of only the three moments of order $k$, $(k+m)/2,$ and $m$.
For $k=1$, only the two moments $p_{(m+1)/2}$ and $p_m$ need to be measured, since $p_1=1$ by normalization.
The restriction that three moments $p_k$, $p_l$, and $p_m$ can only be compared when $l=(k+m)/2$ is somewhat unsatisfactory.
As we will show in Thm.~\ref{thm:ppt_relaxation_hoelder}, this limitation can be lifted.
The proof of our generalization is embarrassingly simple, requiring only a single application of Hölder's inequality.
Nevertheless, we consider it valuable to highlight this result due to its simplicity, its potential practical implications, and the fact that, to the best of our knowledge, it has been overlooked in the existing literature.

The remainder of our work is structured as follows.
In Sec.~\ref{sec:klm_ppt_criterion}, we present the $(k,l,m)$-PPT criterion advertised above.
In Sec.~\ref{sec:quantitative}, we show how yes/no-formulations of entanglement criteria can be turned into quantitative lower bounds on the negativity of a quantum state.
In Sec.~\ref{sec:stieltjes_equiv_to_ppt}, we derive a sufficient condition under which the Stieltjes-$m$ criterion is as strong as the abstract PPT criterion.
In Sec.~\ref{sec:ghz_thresholds}, we compare the power of different entanglement criteria using the example of locally depolarized Greenberger–Horne–Zeilinger states. 
Finally, in Sec.~\ref{sec:new_enumerators}, we introduce a new notion of quantum weight enumerators governing the decay of $p_k$ under local depolarization, and illustrate this framework using the absolutely maximally entangled state on six qubits in Sec.~\ref{sec:ame_thresholds}.

\section{A simple three-moment relaxation of the PPT criterion}
\label{sec:klm_ppt_criterion}

\begin{table*}[t]
\centering
\caption{\add 
Comparison of the moment requirements for different entanglement criteria.
Here, $p_k= \Tr[(\rho^\Gamma)^k]$ denotes the $k$-th PT moment.
The column ``Max.\ copy number'' refers to the largest number of copies required
in a sample-efficient multicopy measurement. 
The last column counts the number of distinct PT moments that must be estimated, excluding the normalization $p_1=1$.
}
\label{tab:overhead}

\setlength{\tabcolsep}{10pt}

\begin{tabular}{lccc}
\hline\hline
\add Criterion & \add Required information & \add Max. copy number & \add No. of PT moments \\
\hline
\add Fidelity witness & \add target-state fidelity & \add $1$ & \add $0$ \\
\add Purity criterion & \add $\Tr[\rho^2]$ and marginal purities & \add $2$ & \add $1$ \\
\add $p_3$-PPT / Stieltjes-3 & \add $p_2,p_3$ & \add $3$ & \add $2$ \\
\add Descartes-$m$ & \add $p_2,\ldots,p_m$ & \add $m$ & \add $m-1$ \\
\add Stieltjes-$m$ & \add $p_2,\ldots,p_m$ & \add $m$ & \add $m-1$ \\
\add $(k,l,m)$-PPT & \add $p_k,p_l,p_m$ & \add $m$ & \add $3$ \\
\add $(1,l,m)$-PPT & \add $p_l,p_m$ & \add $m$ & \add $2$ \\
\hline\hline
\end{tabular}
\end{table*}

As a warm-up, let us rederive our $2\times2$ relaxation of the Stieltjes-$m$ criterion in Eq.~\eqref{eq:hankel_minor_criterion} from scratch.
Assume that $\rho^\Gamma$ is PSD.
Applying the Cauchy-Schwarz inequality, $\Tr[K^\dagger M]^2 \le \Tr[K^\dagger K] \Tr[M^\dagger M]$ to the two matrices $K=(\rho^\Gamma)^{k/2}$ and $M=(\rho^\Gamma)^{m/2}$ yields $p_{(k+m)/2} ^2 \le p_k p_m$, which finishes the proof.
Note that the Cauchy-Schwarz inequality is a special case (with $r=s=2$) of Hölder's inequality~\cite[Eq.~(IV.43)]{bhatia_matrix_analysis_1997}, which states that
\begin{align} \label{eq:hoelder_general}
    \Tr\left [ ((KM)^\dagger (KM))^{1/2} \right]
    &\le \Tr\left[(K^\dagger K)^{{r/2}}\right]^{1/r} 
    \\ &\times 
     \Tr\left[(M^\dagger M)^{{s/2}}\right]^{1/s} \nonumber
\end{align}
for all $ 1\le r,s \le \infty$ with $1/r+1/s=1$ and all matrices $K$ and $M$.
This suggests the following generalization of Eq.~\eqref{eq:hankel_minor_criterion}.

\begin{theorem}  [The $(k,l,m)$-PPT criterion] \label{thm:ppt_relaxation_hoelder} 
    Let $k<l<m$ be integers and define $x= (m-l)/(m-k)$.
    If $p_{l} > p_{k}^x p_{m}^{1-x}$, then $\rho$ is NPT entangled. 
\end{theorem}
\begin{proof} 
{\add We prove the contraposition.
Assume that $\rho^\Gamma $ is PSD.
Then, the matrix $(\rho^\Gamma)^{a}$
is well defined and Hermitian for every choice of $a\ge 0$.
Let $r=1/x$ and  $s = 1/(1-x)$, so that  $1/r+1/s = 1$.
By construction, we have $x\in(0,1)$, which implies $ 1\le r,s \le \infty$, so the conditions for 
Inequality~\eqref{eq:hoelder_general}  are satisfied.
Inserting  $K=(\rho^\Gamma)^{kx}$ and $M=(\rho^\Gamma)^{m(1-x)}$  yields
\begin{align}  \label{eq:hoelder_proof}
\Tr[KM]
\le 
\Tr[K^r]^{1/r} \Tr[M^s]^{1/s}.
\end{align}
We find $\Tr[KM] = \Tr[(\rho^\Gamma)^{kx + m(1-x)}] = \Tr[(\rho^\Gamma)^l] = p_l$.
Similarly, the factors in the right-hand side of Ineq.~\eqref{eq:hoelder_proof}
can be rewritten as 
$\Tr[K^r]^{1/r} %= \Tr[(\rho^\Gamma)^{k}]^x 
= p_k^x$ and 
$\Tr[M^s]^{1/s} %= \Tr[(\rho^\Gamma)^{m}]^{1-x}
=p_m^{1-x}$.
This establishes $p_{l} \le p_{k}^x p_{m}^{1-x}$ whenever $\rho^\Gamma$ is PSD, which finishes the proof.}
\end{proof}

{\add 
Table~\ref{tab:overhead} summarizes the implementation overhead of the
criteria considered in this work. 
The main experimental costs are the largest required copy number and the number of distinct PT moments  that must be
estimated.
The Descartes and Stieltjes criteria require $p_2,\ldots,p_m$.
In contrast, the $(k,l,m)$-PPT criterion uses only $p_k,p_l,p_m$.
Thus, it reduces the number of required PT-moment estimates from $m-1$ to at most three.
}

\newcommand{\lognegativity}{E_\mathcal{N}}

\section{Quantifying entanglement via PT moments}
\label{sec:quantitative}

While the emphasis of this work is on refining entanglement criteria as such, 
we show in this section that knowledge about PT moments also leads to quantitative estimates of entanglement.
This reasoning follows the intuition that if an entanglement criterion is strongly violated, then the underlying entanglement should also be significant in a quantitative sense. 
Such quantitative entanglement witnesses~\cite{quant-ph/0607167, Audenaert06} often provide lower bounds on entanglement measures from experimental data in a substantially more efficient manner than is achievable via quantum state tomography.

In the context of PT moments, previous work~\cite{tarabunga_quantifying_mixed_2025} showed that it is possible to derive lower bounds for the logarithmic negativity,
\begin{align} \label{eq:lognegativity_definition}
\lognegativity[\rho] = \log_2\left( \| \rho^\Gamma \|_1 \right),
\end{align}
where 
$\|A\|_q = \Tr[ (A^\dagger A)^{q/2}]^{1/q}$ 
denotes the Schatten-$q$ norm of a matrix $A$ for $q\ge1$.
Note that the logarithmic negativity is a (non-convex) entanglement measure that upper bounds the distillable entanglement~\cite{Volume,PhD,VidalNegativity,plenio_logarithmic_negativity_2005}.
In analogy with $p_m = \Tr \left[(\rho^\Gamma)^m\right]$, one defines
\begin{align} \label{eq:tilde_moments}
\tilde{p}_m = \| \rho^\Gamma \|_m^m \, .
\end{align}
By mapping the problem to a classical probability distribution and leveraging the monotonicity of Rényi entropies, one can show~\cite{tarabunga_quantifying_mixed_2025} that
\begin{align} \label{eq:old_bound_lognegativity}
\lognegativity[\rho] \ge {\add \frac{1}{2 - m} \log_2(\tilde{p}_m) }+ \frac{1 - m}{2-m}\log_2(p_2)  \, .
\end{align}
If $m$ is even, then $\tilde{p}_m = p_m$, and the lower bound in Eq.~\eqref{eq:old_bound_lognegativity} can be efficiently estimated using the $m$-cycle test.
Our next result generalizes this lower bound, which is recovered in the special case of $l = 2$.

\begin{theorem}[Quantitative lower bounds] 
\label{thm:quantitative_bounds}
    Let $l$ and $m$ be integers and define $x = (m-l)/(m-1)$.
    Then
    \begin{align}
        \lognegativity[\rho] \ge 
        {\add \frac{x-1}{x}  \log_2(\tilde{p}_m) + \frac{1}{x}
        \log_2(\tilde{p}_l)  \,.
        }
    \end{align}
\end{theorem}
\begin{proof}
Replacing $\rho^\Gamma$ in the proof of Thm.~\ref{thm:ppt_relaxation_hoelder} by the matrix $|\rho^\Gamma| = \sqrt{(\rho^\Gamma)^\dagger \rho ^\Gamma}$, which is positive for all states $\rho$,
and specializing to the case $k=1$
yields $\tilde{p}_l \le \tilde{p}_1^x \tilde{p}_m^{1-x}$.
{\add Taking the logarithm shows $\log_2(\tilde{p_l}) \le x\log_2(\tilde p_1) + (1-x) \log_2(\tilde p_m)$.
After dividing by $x>0$ and rearranging terms, we find  
\begin{equation}
\log_2(\tilde{p}_1) \ge \frac{x-1}{x}  \log_2(\tilde{p}_m) + \frac{1}{x} \log_2(\tilde{p}_l) \, ,
\end{equation}
and} 
$\lognegativity[\rho] = \log_2(\tilde{p}_1)$ finishes the proof.
\end{proof}

{\add 
For the sake of concreteness and conciseness, we focus primarily on the PPT criterion and its relaxations throughout this work.
Nevertheless, we should point out that Thm.~\ref{thm:quantitative_bounds}  straightforwardly extends to the so-called index permutation criterion~\cite{horodecki_separability_of_2006}.
This criterion  encompasses the computable cross-norm, also known as the realignment (CCNR) criterion~\cite{chen_a_matrix_2003,rudolph_further_results_2005} as an important special case.
Given an $n$-qubit state $\rho$ and a fixed permutation $\pi\in\text{Sym}(2n)$,
one defines the matrix $\rho^\pi \in \mathbb{C}^{2^n \times 2^n }$ 
through 
$ \bra{a_1,\ldots, a_n} \rho^\pi \ket{a_{n+1}, \ldots, a_{2n}} 
    =
    \bra{a_{\pi^{-1}(1)},\ldots, a_{\pi^{-1}(n)}} 
    \rho
    \ket{a_{\pi^{-1}(n+1)}, \ldots, a_{\pi^{-1}(2n)}}$.
As an example, let $A\subseteq{1,\ldots,n}$ be a subset of qubits and define the permutation
$\Gamma = \prod_{i\in A} (i\quad n+i)$.
Then $\rho^\Gamma$ coincides with the partial transpose with respect to the qubits in $A$, as defined in Eq.~\eqref{eq:partial_transpose_definition}.
As a generalization of Eq.~\eqref{eq:tilde_moments}, one defines the \emph{permutation moments},
\begin{align}
    \tilde p ^\pi_{m} =  \| \rho^\pi  \|_m^m \, ,
\end{align}
for each $m\in \mathbb{N}$.
Whenever the moment order $m$ is even, it is possible to efficiently estimate $\tilde{p}^\pi_m$ by measuring $\rho^{\otimes m}$ using the $\pi$-permutation test~\cite{liu_detecting_entanglement_2022}.
For $m=1$, we define the logarithmic $\pi$-negativity
\begin{align}
    \lognegativity ^\pi [\rho] = \log_2( \tilde{p}^\pi_1) \, ,
\end{align}
which generalizes Eq.~\eqref{eq:lognegativity_definition}.
Replacing $\rho^\Gamma$ by $\rho^\pi$ in the proof of Thm.~\ref{thm:quantitative_bounds} establishes the lower bound
\begin{align} \label{eq:quantitative_bound_generalization}
    \lognegativity ^\pi [\rho]  \ge \frac{x-1}{x} \log_2(\tilde p_m^\pi ) + \frac{1}{x} \log_2(\tilde p_l^\pi ) \, ,
\end{align}
where $l< m$ are integers and $x=(m-l)/(m-1)$.
Note that the special case in which $l=2$ and $\pi$ corresponds to the CCNR criterion   has recently been established  in Ref.~\cite{tarabunga_quantifying_mixed_2025}.

}

\section{The Stieltjes criterion is equivalent to the PPT criterion}
\label{sec:stieltjes_equiv_to_ppt}

In Ref.~\cite{neven_symmetry_resolved_2021},
it has been pointed out that for every $n$-qubit state $\rho$ that is NPT entangled across a bipartition $A\vert B$,
there exists an order $3\le m \le 2^n$ for which the Descartes-$m$ criterion detects this entanglement.
On the other hand, the question of whether the Stieltjes-$m$ criterion (for a sufficiently large value of $m$) is as strong as the PPT criterion  has, to our knowledge, not been addressed before. 
Our next result answers this in the affirmative.

\begin{theorem}
%[Detecting all NPT entangled states with the Stieltjes criterion]
[The Stieltjes criterion detects all NPT states]
\label{thm:stieltjes_iff_ppt}
    Denote by $N$ the number of distinct eigenvalues of $\rho^{\Gamma}$.
    If $\rho $ is NPT entangled and $m\ge 2N-1$, then $H_m$ is not PSD.
    %In other words, the Stieltjes criterion detects all NPT entangled states.
\end{theorem}
\begin{proof}
Denote the eigenvalues of $\rho^\Gamma$ by $\lambda_1,\ldots, \lambda_N$ and their multiplicities by $\mu_1,\ldots, \mu_N$.
Given $\lambda_1<0$, we need to construct $\mathbf{c}\in\RR^M$, with $\mathbf{c}^\text{T} H_m \mathbf{c}<0$,
where $M={(m+1)/2}$ is the dimension of the Hankel matrix $H_m$ from  Eq.~\eqref{eq:stieltjes_definition}.
Consider the polynomial $f(x) = \prod_{i=2}^N (x-\lambda_i)$
and denote its coefficients by $c_0,\ldots, c_{N-1}$. 
By assumption, $M \ge N$, so we can define the vector $\mathbf{c}= (c_0, \ldots, c_{N-1}, 0, \ldots, 0) \in \RR^M$. 
With this choice, we obtain
\begin{align}
    \mathbf{c}^\text{T} H_m \mathbf{c} =  \sum_ {i,j=0}^{N-1} c_i c_j p_{i+j+1} \,.
\end{align}
Since $p_{k} = \sum_{l=1}^N \mu_l \lambda^k_l $, we can further write
\begin{align} \label{eq:stieltjes_ppt_proof_big_sum}
    \mathbf{c}^\text{T} H_m \mathbf{c} =  \sum_{l=1}^N \mu_l \lambda_l \sum_ {i,j=0}^{N-1} c_i\lambda^i %\sum_ {j=0}^{N-1}
    c_j \lambda^j = \sum_{l=1}^N \mu_l \lambda_l f(\lambda_l)^2 \, .
\end{align}
By construction of $f$, we have $f(\lambda_l)= 0$ for all $l \ge 2$.
Hence, only one term in Eq.~\eqref{eq:stieltjes_ppt_proof_big_sum} survives, namely
$ \mathbf{c}^\text{T} H_m \mathbf{c} =  \mu_1 \lambda_1 f(\lambda_1)^2$.
Since the eigenvalues $\lambda_1,\ldots,\lambda_N$ were chosen to be distinct, we have $f(\lambda_1) \neq 0$, which implies $f(\lambda_1)^2 > 0$.  
Finally, we have $\lambda_1 <0 $ with multiplicity $\mu_1 \ge 1$, 
which implies $ \mathbf{c}^\text{T} H_m \mathbf{c}<0$ and finishes the proof.
\end{proof}

{\add
Note that a different problem that is similar in spirit with 
Thm.~\ref{thm:stieltjes_iff_ppt} has been solved before in
Lem.~4 of Ref.~\cite{yu_optimal_entanglement_2021}.
Here, the authors showed that---for an arbitrary moment order $m$---the matrix $H_m$ cannot be PSD unless the moment vector $(p_1,\ldots,p_m)$ defining it can be approximated arbitrarily well by moments of a sequence of PPT states.
This does, however, not imply that for every NPT entangled state $\rho$,
the moment matrix $H_m$ is non-PSD.
In fact, this statement is incorrect unless one selects $m$ sufficiently large.
By our Thm.~\ref{thm:stieltjes_iff_ppt}, it always suffices to choose $m=2^{n+1}-1$ for an unknown $n$-qubit state.
In summary, Lem.~4 of Ref.~\cite{yu_optimal_entanglement_2021}
establishes an asymptotic moment-feasibility statement, whereas
Thm.~\ref{thm:stieltjes_iff_ppt} gives a finite-order detection
guarantee for arbitrary NPT states.
}

Let us apply Thm.~\ref{thm:stieltjes_iff_ppt} to show that the Stieltjes-5 criterion is as strong as the PPT criterion in the simple case of a globally depolarized stabilizer state.
{\add An $n$-qubit stabilizer state vector $\ket{\psi}\in (\CC^2)^{\otimes n}$ is defined as the unique (up to an irrelevant global phase) $+1$-eigenvector of all operators $S\in \mathcal{S}$, where $\mathcal{S}$ denotes the stabilizer group of $\ket{\psi}$.
Hereby, $\mathcal{S}$ must be an Abelian subgroup of the $n$-qubit Pauli group with $\#\mathcal{S}=2^n$ and $-\mathbbm 1 \not \in \mathcal{S}$.}
Let $\Psi _\text{stab} = \ket{\psi}\!\bra{\psi}$ be {\add an arbitrary}  $n$-qubit stabilizer state and 
$A\vert B$ an arbitrary bipartition of $\{1,\ldots,n\}$.
Then, there exist unitaries $U_A$ and $U_B$ acting on $A$ and  $B$, respectively,
such that
$U_A\otimes U_B \ket{\psi} = \ket{\Phi^+}^{\otimes r} \otimes \ket{0}^{n-2r}$.
Hereby, the number $r$ of maximally entangled Bell pairs, $\ket{\Phi^+} =  (\ket{0,0}+\ket{1,1}) / \sqrt{2}$,
that are shared between $A$ and $B$,
can be efficiently computed from the stabilizer group~\cite{fattal_entanglement_in_2004}.
This shows that the Schmidt rank of $\ket{\psi}$ across the bipartition $A|B$ is $2^r$ and all Schmidt coefficients are given by
$  2^{-r}$ with degeneracy $2^r$.
Therefore, the eigenvalues of $\Psi_\text{stab}^\Gamma$ are given by
$\pm 2^{-r}$ with multiplicities $\mu_\pm = 2^{r-1}(2^r \pm1)$, 
and $0$ with multiplicity $\mu_0 = 2^n-2^{2r}$.
For the globally depolarized state $\mathcal{D}^\text{glob}_{\varepsilon}[\Psi_\text{stab}] = (1-\varepsilon) \Psi_\text{stab} +  \varepsilon \frac{\mathbbm {1}}{2^n}$,
the eigenvalues of  $\mathcal{D}^\text{glob}_{\varepsilon}[\Psi_\text{stab}^\Gamma]$ follow as 
$\lambda_\pm = \varepsilon/2^{n} \pm (1-\varepsilon)/2^{r}$ and $\lambda_0=\varepsilon/2^{n}$.
The eigenvalue $\lambda_- $ is negative iff $\varepsilon < 
1-1/(2^{n-r}+1)$, in which case the noisy stabilizer state is NPT entangled.
By virtue of Thm.~\ref{thm:stieltjes_iff_ppt}, it suffices to consider only the moments of order $m\le 5$ to recover this noise threshold.

\begin{corollary}[Low moments suffice for stabilizer states with global noise]
\label{cor:p5_ppt_suffices_for_global_noise}
The Stieltjes-5 criterion % from Eq.~\eqref{eq:stieltjes_definition}
detects entanglement in $\mathcal{D}^\mathrm{glob}_{\varepsilon}[\Psi_\mathrm{stab}]$  for all values of $\varepsilon<1-1/(2^{n-r}+1)$.
\end{corollary} 
\begin{proof}
Having already established the spectrum of $ \mathcal{D}^\text{glob}_{\varepsilon}[\Psi]^\Gamma$, we know that there are only $N=3$ eigenvalues. 
As such, Thm.~\ref{thm:stieltjes_iff_ppt} applies with $m=5$.
\end{proof}

We have just established in Cor.~\ref{cor:p5_ppt_suffices_for_global_noise} that the case of global white noise is, as usual, rather unspectacular. 
As such, let us turn to the more interesting setting of local noise.
We define the single-qubit depolarizing channel, $\mathcal{D}_\varepsilon$, of error strength $\varepsilon$ by $\mathcal{D}_\varepsilon [\rho ] = (1-\varepsilon)\rho + \varepsilon \Tr[\rho]\mathbbm{1}/2$.
As an illustrative example, consider the $n$-qubit 
\emph{Greenberger-Horne-Zeilinger} (GHZ) state, 
whose state vector is defined as
\begin{equation}
\ket{\text{GHZ}^n} = \frac{\ket{0}^{\otimes n} + \ket{1}^{\otimes n}}{\sqrt2} \, .
\end{equation} 
Such GHZ states have been used to probe the quantum nature of large-scale superconducting quantum computers involving 120 qubits in the recent past~\cite{javadiabhari_big_cats_2025}.
We denote the corresponding projector by $\Psi_\text{GHZ}= 
\ket{\text{GHZ}^n}\!\bra{\text{GHZ}^n}$ whenever the number of qubits is clear from context.
For simplicity, we assume that the number of qubits $n$ is even and we consider a balanced bipartition $A = \{1, \ldots, n/2\}$ and $B = \{n/2 + 1, \ldots, n\}$.
For the abstract PPT criterion, early work~\cite{simon_robustness_of_2002} showed that GHZ states exhibit a remarkably high noise threshold,
\begin{align} \label{eq:eps_max_ppt_ghz}
    \varepsilon_\text{max}^\text{PPT} = 1-(2^{2-2/n}+1)^{-1/2} \, .
\end{align}
In our situation, the PPT criterion is known to be both necessary and sufficient for entanglement~\cite{kay_optimal_detection_2011}.
In other words, $\mathcal{D}^{\otimes n}_{\varepsilon}[\Psi_\mathrm{GHZ}]$  is entangled across $A\vert B$ iff $\varepsilon<\varepsilon^\text{PPT}_\text{max}$.
We can invoke Thm.~\ref{thm:stieltjes_iff_ppt} to establish that the Stieltjes criterion fully recovers this threshold without requiring exponentially many moments of $\rho^\Gamma$.

\begin{corollary}[Stieltjes for locally depolarized GHZ states]
\label{cor:stieltjes_for_locally_depolarized_ghz_states}

For $m=2n+5$, the Stieltjes-$m$ criterion detects entanglement in  $\mathcal{D}^{\otimes n}_{\varepsilon}[\Psi_\mathrm{GHZ}]$  for all values of $\varepsilon<\varepsilon_\mathrm{max}^\mathrm{PPT}$.
\end{corollary} 
\begin{proof}
Since we assume that the number of qubits $n$ is even and that the bipartition $A \vert B$ is balanced, we can directly use the eigenvalue spectrum of $\mathcal{D}_\varepsilon^{\otimes n}[\Psi_\mathrm{GHZ}^{\Gamma}]$ as reported in the literature~\cite{liu_decay_of_2009}.
The eigenvalues are given by 
\begin{align} \label{eq:eigenvalues_pt_ghz_1}
    \lambda_{j} = \frac{1}{2} \left( \left(1-\frac{\varepsilon}{2}\right)^{n-j} \left(\frac{\varepsilon}{2}\right)^j + \left(1-\frac{\varepsilon}{2}\right)^{j} \left(\frac{\varepsilon}{2}\right)^{n-j} \right)
\end{align}
with multiplicities $\mu_j = \binom{n}{j}-2\delta_{j, n/2}$,
as well as
\begin{align} \label{eq:eigenvalues_pt_ghz_2}
    \lambda_\pm = \left(1-\frac{\varepsilon}{2}\right)^{n/2} \left(\frac{\varepsilon}{2}\right)^{n/2}   \pm \frac{(1-\varepsilon)^n}{2}
\end{align}
with multiplicities $\mu_\pm = 1$.
Thus, we have no more than $N=n+3$ distinct eigenvalues, 
and Thm.~\ref{thm:stieltjes_iff_ppt} finishes the proof.
\end{proof}

Note that the moment bound in Cor.~\ref{cor:stieltjes_for_locally_depolarized_ghz_states} is not tight.
For small system sizes, we find that significantly lower moments already suffice to fully reproduce the PPT criterion: for $n=2$, $4$, and $6$, this is achieved with $m=3$, $9$, and $11$, respectively, instead of the general bound $2n+5 = 9$, $13$, and $17$ predicted by the corollary.

\section{Noise thresholds for GHZ states}
\label{sec:ghz_thresholds}

In the previous sections, we introduced the $(k,l,m)$-PPT criterion (Thm.~\ref{thm:ppt_relaxation_hoelder}) and established that the Stieltjes-$m$ criterion is as strong as the PPT criterion  
whenever the partially transposed state has no more than $(m+1)/2$ distinct eigenvalues (Thm.~\ref{thm:stieltjes_iff_ppt}).
However, the precise relationship between these and other common entanglement criteria remains unclear.
In this section, we address this question by investigating the important case of locally depolarized GHZ states. 
Importantly, applying each criterion (except for the fidelity criterion) does not rely on prior knowledge of the exact form of the state.

As in Cor.~\ref{cor:stieltjes_for_locally_depolarized_ghz_states}, we restrict attention to an even number of qubits, $n$, and consider the balanced bipartition $A=\{1,\ldots,n/2\}$ versus $B=\{n/2+1, \ldots, n\}$.
As a figure of merit (for comparing different criteria), we consider the maximal amount, $\varepsilon_{\max}$, of local depolarizing noise such that a given criterion still detects entanglement in the noisy state $\mathcal{D}_\varepsilon^{\otimes n}[\Psi_{\text{GHZ}}]$.
Since GHZ states are well-studied in the literature, we begin our discussion by reviewing various entanglement criteria as well as the experimental state-of-the-art.

\begin{figure*}
    \centering
    \includegraphics[width= .8\linewidth]{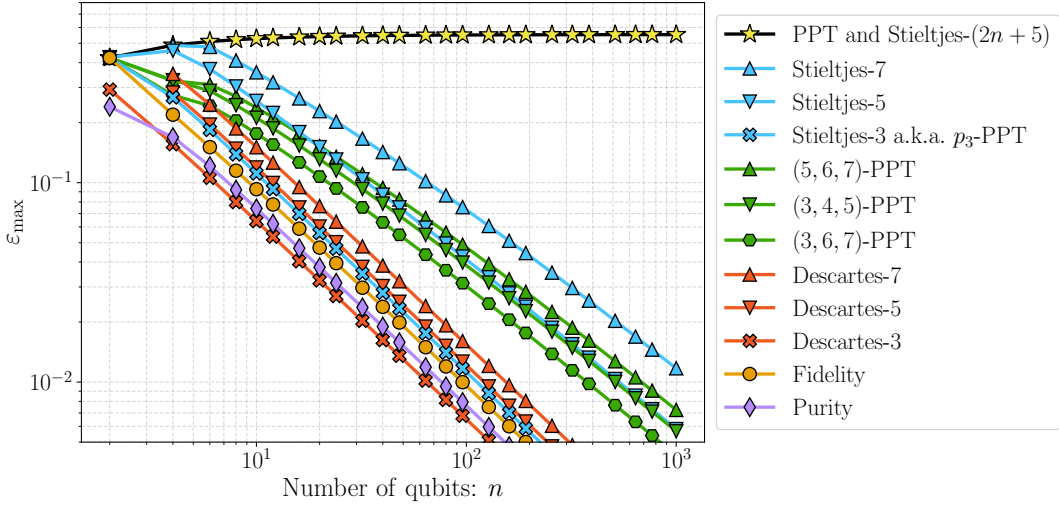}
    \caption{Detectability thresholds for various entanglement criteria applied to locally depolarized GHZ states as a function of the number $n$ of qubits. If the error rate is below $\varepsilon_\text{max}$, the corresponding criterion certifies entanglement across a balanced bipartition.
    The abstract PPT criterion yields the highest thresholds, however, it is not experimentally accessible if $n$ is large. 
    By virtue of Cor.~\ref{cor:stieltjes_for_locally_depolarized_ghz_states}, the PPT threshold is already achieved by the Stieltjes-$m$ criterion for $m=2n+5$.
    Fidelity and purity criteria have been demonstrated experimentally with 
    $120$ and six qubits, respectively, due to their comparatively modest implementation overhead. 
    However, their thresholds scale inversely with the number of qubits, which limits their applicability. 
    If $m$ is constant, the Descartes-$m$ criterion does not overcome this limitation, despite making use of higher PT moments.
    The Stieltjes-$m$ criterion and the $(k,l,m)$-PPT criterion from Thm.~\ref{thm:ppt_relaxation_hoelder}, on the other hand, 
    converge more slowly to zero for $m\in\{5,7\}$, possibly enabling experimental entanglement detection with a larger number of qubits in the near future.  
    The $(k,l,m)$-PPT criterion is easier to implement than the Stieltjes-$m$ criterion, as it requires access to fewer PT moments. }
    \label{fig:ghz_loc_depol_vs_n}
\end{figure*}

Recall from Eq.~\eqref{eq:eps_max_ppt_ghz} that the noisy state $\mathcal{D}_\varepsilon^{\otimes n}[\Psi_{\text{GHZ}}]$ is entangled iff $\varepsilon < \varepsilon^\text{PPT}_\text{max} = 1-(2^{2-2/n}+1)^{-1/2}$.
In Fig.~\ref{fig:ghz_loc_depol_vs_n}, we plot this PPT noise threshold (yellow stars) 
as a function of $n$.
As noted before~\cite{simon_robustness_of_2002}, 
we can see that $\varepsilon^\text{PPT}_\text{max}$ converges to $1-1/\sqrt{5} \approx 0.5528$ in the large-$n$ limit.
Therefore, GHZ states can withstand an impressive amount of more than $50\,\%$ of effective per-qubit noise before its entanglement is lost.
However, applying the abstract PPT criterion becomes experimentally intractable as the number of qubits increases.
In principle, Cor.~\ref{cor:stieltjes_for_locally_depolarized_ghz_states} resolves this issue, as it suffices to measure the moments $p_k = \Tr[(\rho^\Gamma_\text{GHZ})^k]$ for all $ k\le 2n+5$ in order to achieve $\varepsilon^{\text{Stieltjes-}2n+5}_\text{max} = \varepsilon^\text{PPT}_\text{max}$.
That said, measuring all these moments for systems with hundreds of qubits would be extremely challenging.

In practice, \emph{fidelity measurements} are most commonly used since they can be performed in single-copy experiments with polynomial sample complexity~\cite{flammia_direct_fidelity_2011}.
State-of-the-art experiments have demonstrated GHZ states across a range of platforms and system sizes.
In photonic systems, GHZ states of up to $18$ qubits encoded into $6$ photons have been realized with a fidelity of $70.8 \pm 1.6\,\%$~\cite{wang_18_qubit_2018}, and $14$-photon GHZ states have been demonstrated with a fidelity of $76 \pm 6\,\%$~\cite{thomas_efficient_generation_2022}.
In trapped-ion systems, GHZ states of up to $32$ qubits have been realized with a fidelity of $82 \pm 1 \,\%$~\cite{moses_a_race_2023}.
In cold-atom platforms, a fidelity of $54.2 \pm 1.8\,\%$ %was 
has been reported for $20$ physical qubits~\cite{omran_generation_of_2019}.
Using a similar platform, GHZ states of $4$ logical qubits have been realized with fidelities of $72 \pm 2\,\%$ using deterministic error correction and up to $99.85^{+0.1}_{-1.0}\,\%$ using probabilistic error detection~\cite{bluvstein_logical_quantum_2024}.
Finally, in superconducting qubit systems, GHZ states have been realized with a fidelity of $51.9 \pm 1.4\,\%$ for $32$ qubits using deterministic preparation~\cite{kam_characterization_of_2024}, and with up to $120$ qubits achieving a fidelity of $56 \pm 3\,\%$ using probabilistic error-detection techniques~\cite{javadiabhari_big_cats_2025}.
On the theoretical side, one can show~\cite{tanizawa_simplest_fidelity_2023} that the fidelity, $F_\text{GHZ}(\varepsilon) = \bra{\text{GHZ}^n} \mathcal{D}_\varepsilon^{\otimes n}[\Psi_\text{GHZ} ] \ket{\text{GHZ}^n}$,
of the locally depolarized GHZ state is given by
% \begin{eqnarray}%
% &&\bra{\text{GHZ}^n} \mathcal{D}_\varepsilon^{\otimes n}[\Psi_\text{GHZ} ] \ket{\text{GHZ}^n}\nonumber\\
% &=& \tfrac{1}{2}\left[\left(1 -\tfrac{\varepsilon}{2}\right)^n  + \left(1 -\varepsilon\right)^n -  \left(\tfrac{\varepsilon}{2}\right)^n \right] \, .
% \end{eqnarray}
\begin{align}
    F_\text{GHZ}(\varepsilon)
= \tfrac{1}{2}\left[\left(1 -\tfrac{\varepsilon}{2}\right)^n  + \left(1 -\varepsilon\right)^n -  \left(\tfrac{\varepsilon}{2}\right)^n \right]
\, . 
\end{align}
Therefore, it decays exponentially for any fixed value of  $\varepsilon > 0$.
In Fig.~\ref{fig:ghz_loc_depol_vs_n}, we plot 
the noise threshold $\varepsilon_\text{max}^\text{fidelity}$ (orange circles),
defined as the value at which the fidelity drops below $50\,\%$.
Unlike $\varepsilon_\text{max}^\text{PPT}$, we observe that the fidelity threshold  scales as $\varepsilon_\text{max}^\text{fidelity}  = O(n^{-1})$ and rapidly converges to zero, reaching $\varepsilon_\text{max}^\text{fidelity} \approx  10^{-2}$ at $n=100$ qubits.
We stress the consistency of this simple noise model with experimental state-of-the-art demonstrations, which exhibit limiting error rates and maximum system sizes on the order of $\varepsilon=10^{-2}$ and $n=100$, respectively~\cite{javadiabhari_big_cats_2025}.

% The GHZ state is further distinguished by maximizing the $n$-body Shor--Laflamme quantum weight enumerator~\cite{shor_quantum_analog_1997, tran_quantum_entanglement_2015, eltschka_maximum_nbody_2020},
% \begin{align} \label{eq:n_body_sector_length}
%     a_n[\rho]  = \frac{1}{2^n} \sum_{P \in \{X,Y,Z\}^{\otimes n}} \Tr[\rho P ]^2 \, ,
% \end{align}
% attaining $a_n[\Psi_\text{GHZ}] = 0.5 + 2^{-n}\delta_{n,\text{even}}$.
% Since every state $\rho$ with $a_n[\rho]>2^{-n}$ is entangled~\cite{kloeckl_characterizing_multipartite}, 
% this leads implies a  noise threshold $\varepsilon_\text{max}^{n\text{-SL}} = 1 - 2^{-\tfrac{1}{2}+\tfrac{1}{2n}}$ if $n$ is odd, and similarly if $n$ is even~\cite{miller_shor_laflamme_2023}.
% As shown in Fig.~\ref{fig:ghz_loc_depol_vs_n}, the $n$-body criterion (pink pentagons) is weaker than the PPT criterion.
% Still, it converges to a nonzero threshold, namely $1 - 1/\sqrt{2} \approx 0.2929$.
% In a recent two-copy experiment employing pairwise Bell measurements, a value of $a_n[\rho] = 0.32 \pm 0.02$ was measured for $n = 6$ qubits,
% which certifies entanglement~\cite{miller_experimental_measurement_2024}.
% However, as the number of qubits increases, experimentally estimating Eq.~\eqref{eq:n_body_sector_length} becomes increasingly demanding for GHZ states, requiring exponentially many samples to achieve a fixed additive precision~\cite{miller_experimental_measurement_2024}.

Experimental implementations of \emph{two-copy purity measurements} using GHZ states have, to our knowledge, only been realized in an early photonic demonstration~\cite{bovino_direct_measurement_2005} and a more recent trapped-ion platform~\cite{miller_experimental_measurement_2024}.
In Ref.~\cite{bovino_direct_measurement_2005}, an experiment with two photonic Bell pairs (i.e., GHZ states on $n=2$ qubits) reported a global purity $\mathrm{Tr}[\rho^2] = 0.90 \pm 0.02$ and a single-qubit subsystem purity $\mathrm{Tr}[\rho_A^2] = 0.43 \pm 0.02$, thereby certifying entanglement.
Similarly, in Ref.~\cite{miller_experimental_measurement_2024}, 
the global purity of a six-qubit GHZ state realized on a trapped-ion device was measured to be
$\Tr[\rho^2] = 0.62\pm0.01$, together with an averaged three-qubit subsystem purity of
% $a'_{5} [\rho ] = 0.39 \pm 0.01$.
\begin{align} \label{eq:apd_experiment}
   \frac{1}{\binom{6}{3}}\sum_{\substack{A\subset\{1,\ldots, 6\}\\ \#A = 3}} \Tr[\rho_A^2] = 0.43 \pm 0.01 \, ,
\end{align}
where the sum runs over all three-qubit subsystems.
Since the observed subsystem purities are smaller than the global purity, the state is entangled across a bipartition of size  $3\vert 3$.
Next, we show---for general $n$-qubit GHZ states---how to compute the theoretical threshold $\varepsilon_\text{max}^\text{purity}$ below which the purity criterion certifies entanglement across a balanced bipartition.
We leverage the concept of Shor--Laflamme enumerators,
which have their origins in the theory of quantum error correction~\cite{shor_quantum_analog_1997}.
They are defined as 
\begin{align} \label{eq:shor_laflamme_enumerators}
    a_i[\rho] = \frac{1}{2^n} \sum_{\substack{P \in \{I,X,Y,Z\}^{\otimes n}\\ \wt(P) = i}} \Tr[\rho P ]^2 \, ,
\end{align}
for each $i\in\{0,\ldots,n\}$,
where the weight $\wt(P)$ of a Pauli operator $P$ is defined as the number of non-identity tensor factors.
For a stabilizer state $\Psi = \ket{\psi}\!\bra{\psi}$ with stabilizer group $\mathcal{S}$, one has $a_i[\Psi] = \#\{S \in \mathcal{S} \ \vert \ \wt(S) = i\}/2^n$, 
i.e., the probability that a uniformly random stabilizer element has weight $i$. 
In particular, we have $a_i[\Psi_\text{GHZ}] = (\binom{n}{i}\delta_{i,\text{even}} + \delta_{i,n})/2^n$ for the $n$-qubit GHZ state~\cite{aschauer_local_invariants_2004}.
Since Shor--Laflamme enumerators generally decay as $a_i[\mathcal{D}_\varepsilon^{\otimes n}[ \Psi]] = (1-\varepsilon)^{2i} a_i[\Psi]$, 
one can write the global purity of the noisy state as
\begin{align} \label{eq:purity_decay}
    \Tr\left[ (\mathcal{D}_\varepsilon^{\otimes n}[\Psi])^2 \right] = \sum_{i=0}^n  a_i[\Psi] (1-\varepsilon)^{2i} \, .
\end{align}
Similarly, the averaged (over all marginals of size $n/2$) subsystem purity (see Eq.~\eqref{eq:apd_experiment} for an example)
\begin{align}
    a'_{n/2} [\rho ] =   \frac{1}{\binom{n}{n/2}} \sum_{\substack{A\subset\{1,\ldots,n\}\\ \# A = n/2}} \Tr \left[  \rho_A^2 \right] \,,
\end{align}
which is also known as Rains' unitary enumerator~\cite{rains_quantum_weight_1998} of order $n/2$, 
can be expressed through Shor--Laflamme enumerators via
\begin{align}
    a'_{n/2}\left [  \mathcal{D}_\varepsilon^{\otimes n}[\Psi]  \right] = 
\frac{2^{{n}/{2}}}{\binom{n}{{n}/{2}}}
\sum_{i=0}^n     \binom{n-i}{{n}/{2}} a_i[\Psi] (1-\varepsilon)^{2i}  \, .
\end{align}
In Fig.~\ref{fig:ghz_loc_depol_vs_n}, we plot the noise threshold $\varepsilon_\text{max}^\text{purity}$   (purple diamonds)
below which  
\begin{align}
a'_{n/2}\left [  \mathcal{D}_\varepsilon^{\otimes n}[\Psi_\text{GHZ}]  \right]< 
    \Tr\left[ (\mathcal{D}_\varepsilon^{\otimes n}[\Psi_\text{GHZ}])^2 \right] \, .
    \end{align}
Whenever $\varepsilon<\varepsilon_\text{max}^\text{purity}$, the noisy state $\mathcal{D}_\varepsilon^{\otimes n}[\Psi_\text{GHZ}]$ is entangled across each bipartition $A\vert B$ with $\#A=\#B$.
We observe that this purity criterion performs slightly worse than the fidelity criterion, exhibiting the same $O(n^{-1})$ scaling.
Since the fidelity criterion certifies entanglement across all bipartitions, is easier to implement, and has a higher noise threshold, 
our analysis indicates that purity-based approaches may be less suitable for experiments aimed at certifying entanglement across a \emph{balanced} bipartition in GHZ states.
Note, however, that the noise threshold for the purity criterion is much higher than the fidelity threshold (albeit still suffering from an $O(n^{-1})$ scaling) if, instead of the balanced bipartition, a bipartition $A\vert B$ of size $\#A=1$ versus $\#B=n-1$ is probed~\cite{miller_shor_laflamme_2023}.
Similarly, the purity criterion performs significantly better (exhibiting an $O(1)$ scaling) when applied to cluster states instead~\cite{miller_shor_laflamme_2023}.

Having discussed criteria accessible via single- and two-copy experiments (fidelity and purity), we now turn to higher-order PT moments.  
From Eq.~\eqref{eq:eigenvalues_pt_ghz_1} and~\eqref{eq:eigenvalues_pt_ghz_2}, we know the eigenvalues $\lambda_j$ and multiplicities $\mu_j$ of $\mathcal{D}_\varepsilon^{\otimes n}[\Psi_\text{GHZ}^\Gamma]$.  
This allows us to express
\begin{align}
    p_k\!\left[\mathcal{D}_\varepsilon^{\otimes n}[\Psi_\text{GHZ}^\Gamma]\right]
    = \sum_{j \in \{0,1,\ldots,n,\pm\}} \mu_j \lambda_j^k \, ,
\end{align}
and thereby compute the corresponding noise thresholds for higher-order relaxations of the PPT criterion.
We restrict our analysis to $k \le 7$, as measuring $p_k$ becomes significantly more challenging with increasing moment order.  
In Fig.~\ref{fig:ghz_loc_depol_vs_n}, crosses, downward-pointing triangles, and upward-pointing triangles denote the cases where the highest required order is $m = 3$, $5$, and $7$, respectively, while the color indicates the corresponding family of criteria.

We first consider the family of \emph{Descartes criteria}~\cite{neven_symmetry_resolved_2021, bradshaw_a_closed_2025}.
In Fig.~\ref{fig:ghz_loc_depol_vs_n}, we see that
their noise thresholds (red curves) are not able to overcome the $O(n^{-1})$ scaling from which also the fidelity and purity thresholds suffer.
Notably, the Descartes criterion for $m=3$ (red crosses) performs even worse than the purity criterion.
As expected~\cite{neven_symmetry_resolved_2021}, 
the thresholds increase with the maximum moment order $m$.
Indeed, already the Descartes-$5$ criterion can tolerate more noise than the fidelity criterion.

Next, we consider the family of \emph{Stieltjes criteria}~\cite{neven_symmetry_resolved_2021, yu_optimal_entanglement_2021}, and plot their noise thresholds in Fig.~\ref{fig:ghz_loc_depol_vs_n}.
For $m=3$ (blue crosses), the Stieltjes criterion 
is equivalent to the $p_3$-PPT criterion~\cite{elben_mixed_state_2020},
and we see that its noise threshold suffers from the same $O(n^{-1})$ scaling.
However, we also observe that the $p_3$-PPT criterion is slightly stronger than both the fidelity and the purity criterion.
For $m=5$ and $m=7$ (blue triangles), we find a better scaling of
$\varepsilon_\text{max}^\text{Stieltjes} = O(n^{-\alpha})$, where we extract
$\alpha = 0.84\pm0.4$ from the plot.
Notably, the Stieltjes-$7$ criterion can tolerate an effective per-qubit error rate exceeding $10^{-2}$ even for GHZ states with $n=1000$ qubits.
We attribute the stronger performance of the Stieltjes-$m$ compared to the Descartes-$m$ criterion to the fact that the latter tests only $m$ polynomials, whereas the former imposes a larger number of constraints, corresponding to all principal minors.

As discussed in the introduction, a drawback of the Stieltjes-$m$ criterion from an experimental perspective is that it requires knowledge of the PT moments $p_k$ for all $k\le m$.
By contrast, our $(k,l,m)$-PPT criterion from Thm.~\ref{thm:ppt_relaxation_hoelder} requires only $p_k$, $p_l$, and $p_m$.
We compute the noise thresholds for some of the best combinations of $k$, $l$, and $m$
and plot them in Fig.~\ref{fig:ghz_loc_depol_vs_n}  (green curves).
We observe that, for a given value of $m$,
the $(k,l,m)$-PPT thresholds are lower than those of the full Stieltjes-$m$ criterion.
This is not unexpected, as the latter exploits more information.
Nevertheless, the $(k,l,m)$-PPT thresholds retain the $O(n^{-\alpha})$ scaling of the Stieltjes-$m$ criterion for $m\ge5$.
In particular, they outperform the Descartes-$m$ criterion, despite leveraging less information.
Interestingly, we find that taking consecutive moments, e.g., $(k,l,m)=(3,4,5)$ or $(k,l,m)=(5,6,7)$, performs better than, e.g., $(k,l,m)=(3,6,7)$, as shown by the green hexagons in Fig.~\ref{fig:ghz_loc_depol_vs_n}.
Overall, we find that optimal performance is obtained when $k$ and $m$ are odd and $l$ is even.

To illustrate the modest savings offered by the $(k,l,m)$-PPT criterion, consider a device with $n=300$ qubits (per copy) and an effective per-qubit error rate of $\varepsilon=10^{-2}$.
Since the Stieltjes-$5$ and $(3,4,5)$-PPT criteria have nearly identical thresholds of approximately $0.016$, both are capable of detecting entanglement in this scenario. 
However, applying the $(3,4,5)$-PPT criterion does not require measuring the state's purity, $p_2$, thereby reducing the experimental overhead by roughly $25\,\%$.

{\add

Finally, let us emphasize that Fig.~\ref{fig:ghz_loc_depol_vs_n} addresses the abstract problem of comparing different entanglement criteria in a physically motivated setting.
More precisely, it shows up to which noise level each criterion can, in principle, certify entanglement in the idealized scenario of perfect measurements.
Extending this analysis to account for realistic measurement imperfections is beyond the scope of the present work.

%% Local shadows:
% Mathematically, erroneous measurements (modeled by a bitflip channel on the classical measurement outcome) are equivalent to local depolarizing noise.
% Therefore, one cannot distinguish these error mechanisms in an experiment.

}

\section{Weight enumerators for the decay of PT moments under local depolarizing noise}
\label{sec:new_enumerators}

In the previous section, we analyzed the performance of moment-based entanglement criteria in a concrete setting. 
We now take a more general perspective and investigate how PT moments evolve under local depolarizing noise
{\add for arbitrary $n$-qubit states.
Understanding this evolution is important both for analyzing the robustness of moment-based entanglement criteria and for obtaining analytical expressions for PT moments in noisy quantum states.}
To this end, we introduce a family of quantum weight enumerators (QWEs) that capture this decay and naturally extend the well-known Shor--Laflamme QWEs, which were first defined in the context of quantum error correction~\cite{shor_quantum_analog_1997}.

{ \add 
As an intuitive motivation, recall from Eq.~\eqref{eq:purity_decay}
that the purity (i.e., the second moment) of a pure state $\Psi = \ket{\psi}\!\bra{\psi}$  under local depolarizing noise
behaves like a mixture of exponential decays with different rates.
How the initial purity is distributed among the different decay modes is determined by the Shor--Laflamme QWEs. 
The coefficient $a_i[\Psi]$ quantifies the fraction of the purity supported on Pauli components of weight $i$, each of which is suppressed as $(1-\varepsilon)^{2i}$.
Thus, the purity decay is not described by a single decay rate, but by a weighted sum of decay modes indexed by Pauli weight.
We refer the interested reader to the interactive visualization tool ``GraphStateVis'',   where such purity decay effects can be explored for the class of graph states~\cite{miller_graphstatevis_interactive_2021}.

}

{\add We now show that an analogous decomposition exists for partial transpose moments.}
Let $\rho$ be an $n$-qubit state, and let $\mathcal{D}^{\otimes n}_\varepsilon[\rho]$ denote its evolution under local depolarizing noise.
We can write 
\begin{align}
     \mathcal{D}^{\otimes n}_\varepsilon[\rho^\Gamma ] = 
     \tfrac{1}{2^n}
     \hspace{-4mm}
     \sum _{P \in \{I,X,Y,Z\}^{\otimes n}}
     \hspace{-4mm}
     \vartheta(P) \Tr[\rho P] (1-\varepsilon)^{\wt(P)} P \, ,
\end{align}
where for an $n$-qubit Pauli operator, $P=P^{(1)}\otimes\ldots\otimes P^{(n)}$, %$ \in\{I,X,Y,Z\}^{\otimes n}$, 
we define
\begin{align} \label{eq:ppt_signflip_function}
\vartheta(P) = (-1)^{\#\{\, a\in A \  \vert \ P^{(a)}=Y\}} \, .
\end{align}
For multiple Pauli operators, $P_1,\ldots, P_k \in \{I,X,Y,Z\}^{\otimes n}$, 
we similarly define 
$\vartheta(P_1,\ldots,P_k) = \prod_{i=1}^k \vartheta(P_i)$ 
as well as
\begin{align}
    \varphi(P_1,\ldots,P_k) = \frac{1}{2^n} \Tr[\textstyle\prod_{i=1}^k P_i] \in\{0, \pm 1, \pm \imag\} \, .
\end{align}
With this notation at hand, we can write
\begin{widetext}
\begin{align} 
    { \add 
    p_k[\mathcal{D}^{\otimes n}_\varepsilon[\rho]] 
    =
    \Tr\left[((\mathcal{D}^{\otimes n}_\varepsilon[\rho])^{\Gamma})^k\right]
    }
    =  \sum_{P_1,\ldots,P_k \in \{I,X,Y,Z\}^{\otimes n}} \frac{\vartheta(P_1,\ldots,P_k) 
    \varphi(P_1,\ldots, P_k)}{2^{n(k-1)}} 
     \prod_{i=1}^k(1-\varepsilon)^{\wt(P_i)} \Tr[P_i \rho] \, .
     \label{eq:pk_local_without_enums}
\end{align}
\end{widetext}
Stratifying Eq.~\eqref{eq:pk_local_without_enums} by all possible values that $\vartheta(P_1,\ldots,P_k)$, $\varphi(P_1,\ldots,P_k) $, and $\wt(P_i)$ can assume, we finally obtain
\begin{align} \label{eq:pk_local_with_enums}
    p_k[(\mathcal{D}^{\otimes n}_\varepsilon[\rho])^{\Gamma}] 
    = 
    \sum_{\substack{\theta\in\{\pm1\} \\  \phi\in\{\pm 1, \pm \imag\}}} 
    \theta \phi \sum_{w=0}^{nk}(1-\varepsilon)^{w}  c_w^{(k,\theta, \phi)}[\rho] \, ,
\end{align}
where we have introduced the real-valued QWEs
\begin{align} \label{eq:new_enumerator}
     c_w^{(k,\theta, \phi)}[\rho]  = \frac{1}{2^{n(k-1)}} 
     \sum_{\substack{P_1,\ldots, P_k \in  \{I,X,Y,Z\}^{\otimes n}\\ 
     \vartheta(P_1,\ldots,P_k)=\theta \\
     \varphi(P_1,\ldots,P_k)=\phi \\
     \sum_{i=1}^k \wt(P_i)=w 
     }}
      \prod_{i=1}^k\Tr[\rho P_i]
     \, ,
\end{align}
which implicitly, via Eq.~\eqref{eq:ppt_signflip_function}, depend on the bipartition $A\vert B$ with respect to which the partial transpose is taken.
For $k=2$, our QWEs are directly related to the well-known
Shor--Laflamme QWEs from Eq.~\eqref{eq:shor_laflamme_enumerators}
via 
$ c_{w}^{(2,\theta,\phi)} [\rho] = \delta_{\theta,+1}\delta_{\phi,+1} \delta_{w,2i} \times a_i[\rho]$,
regardless of the bipartition.
For $k \ge 3$, the QWEs in Eq.~\eqref{eq:new_enumerator} belong to the general class of local polynomial invariants~\cite{rains_polynomial_invariants_2000, grassl_computing_local_1998}, which is less well understood than the $k=2$ case.

Let us investigate the important special case of stabilizer states in greater detail.
For every $n$-qubit stabilizer state $\Psi_\text{stab}=\sum_{S\in \mathcal{S}}S/2^n$,  
where $\mathcal{S}$ denotes the stabilizer group,
the condition $-\mathbbm 1 \not \in \mathcal{S}$ implies that  $c_w^{(k,\theta, \phi)} = 0$ is trivial whenever $\phi\in\{-1, \pm\imag\}$.
In the opposite case,
we have $c_w^{(k,\theta, +1)} =  C_w^{(k,\theta)}/2^{n(k-1)}$, where
\begin{align}
      \nonumber
    C_w^{(k, \theta)} = \# \Biggl\{& S_2,   
    \ldots, S_{k} \in \mathcal{S} 
    \ \Bigg\vert \  
     \vartheta(\textstyle \prod_{i=2}^{k} S_i , S_2, \ldots, S_{k}) = \theta, \\ 
      &  \wt(\textstyle \prod_{i=2}^{k} S_i)+\sum_{i=2}^{k}\wt(S_i) = w \Biggr\} \, .
\label{eq:new_enumerator_stab} 
\end{align}
%%%% same equation without linebreak
% \begin{widetext}
%     \begin{align}
%     C_w^{(k, \theta)} = \# \Biggl\{& S_2,   
%     \ldots, S_{k} \in \mathcal{S} 
%     \ \Bigg\vert \  
%      \vartheta(\textstyle \prod_{i=2}^{k} S_i , S_2, \ldots, S_{k}) = \theta,  
%        \hspace{2mm}  \wt(\textstyle \prod_{i=2}^{k} S_i)+\sum_{i=2}^{k}\wt(S_i) = w \Biggr\} \, .
% \label{eq:new_enumerator_stab} 
% \end{align}
% \end{widetext}
Therefore, Eq.~\eqref{eq:pk_local_with_enums} simplifies to 
\begin{align} \label{eq:pk_local_with_enums_stab}
    p_k[(\mathcal{D}_\varepsilon^{\otimes n}[\Psi_\text{stab}])] = 
    \sum_{w=0}^{nk} \frac{ C_w^{(k, +)} -  C_w^{(k, -)}}{2^{n(k-1)}}  (1-\varepsilon)^{w} \, .
\end{align}
In the important special case of a trivial bipartition $A\vert B$ with $A=\diameter$,
we have $C_w^{(k,-)} =0$ and, therefore, the moments   of $\rho_\varepsilon=\mathcal{D}_\varepsilon^{\otimes n}[\Psi_\text{stab}]$ can be written as
\begin{align} \label{eq:decay_tr_rho_k}
    \Tr\left[ \rho_\varepsilon^k\right] = \frac{1}{2^{n(k-1)}} \sum_{w=0}^{nk} C_w^{(k,+)} (1-\varepsilon)^w \, .
\end{align}
Using Eq.~\eqref{eq:new_enumerator_stab}, the QWEs $ C_w^{(k, \theta)}$, which appear in Eq.~\eqref{eq:pk_local_with_enums_stab}, can in principle be computed via an exhaustive search over all $(k-1)$-tuples of stabilizers.
However, this approach leads to an algorithm with runtime that is exponential in both $n$ and $k$. 
As we show next, by leveraging the theory of Fourier analysis on finite Abelian groups~\cite{terras_fourier_analysis_1999}, the exponential runtime in $k$ can be avoided.

Consider a fixed stabilizer state with stabilizer group, $\mathcal{S}$, and a fixed bipartition, $A\vert B$, which determines the function $\vartheta$ from Eq.~\eqref{eq:ppt_signflip_function}.
For each $z\in \CC$, we define the two functions $f^\pm_z: \mathcal{S} \rightarrow \CC$, where $ f^+_z(S)= z^{\wt(S)}$ and $f^-_z(S) =\vartheta(S)z^{\wt(S)}$.
Evaluating the $k$-fold discrete convolution of $f_z^\pm$ at the identity element $I\in\mathcal{S}$ yields
\begin{align}
     (f_z^\pm)^{\scalebox{1.0}{$\ast$} k} (I ) = \sum_{w=0}^{nk} (C_w^{(k,+)} \pm C_w^{(k,-)}) z^w \,.
\end{align}
The Fourier transform, $\mathcal{F}$, over the finite group $\mathcal{S}$ turns convolutions into pointwise products, i.e., $\mathcal{F} ((f_z^\pm)^{\scalebox{1.0}{$\ast$} k} ) = (\mathcal{F}(f^\pm_z))^k$, where $\widehat{f}^\pm_z =\mathcal{F}(f^\pm_z) $ are functions from the group of characters, 
\begin{align}
\widehat {\mathcal{S}} = \{ \chi : \mathcal{S} \rightarrow  \{\pm1\} \ \vert \  \chi \, \text{ is a group homomorphism} \} \, ,
\end{align}
to $\CC$.
More precisely,  for each $\chi \in \widehat{\mathcal{S}}$, we have $\widehat f_z^+ (\chi) = \sum_{S \in \mathcal{S}} \chi(S) z^{\wt(S)} $ and  $\widehat f^-_z (\chi) = \sum_{S \in \mathcal{S}} \chi(S) \vartheta(S)z^{\wt(S)} $.
Note that $\vartheta$ itself is not a character (unless $A=\diameter$) because $\vartheta(X_a) = \vartheta(Z_a) = 1 $ but $\vartheta(X_aZ_a) =-1$ for each qubit $a\in A$.
Let 
\begin{equation}
\{\alpha^\pm_1(z),\ldots, \alpha^\pm_{N^\pm}(z) \} = \{ \widehat f^\pm_z (\chi) \ \vert \  \chi \in \widehat{\mathcal{S}} \} 
\end{equation}
be the sets of $N^\pm$ polynomials that occur,
and denote their multiplicities by $m^\pm_i = \# \{\chi \in \widehat{\mathcal{S}} \ \vert \ \widehat{f}^\pm_z(\chi) = \alpha^\pm_i(z)\}$.
The inversion formula, see e.g., Thm.~2 in Chap.~10 of Ref.~\cite{terras_fourier_analysis_1999}, allows us to write 
\begin{align} \label{eq:inversion_formula}
     (f^\pm_z)^{\scalebox{1.0}{$\ast$} k} (S) = \frac{1}{2^n} \sum_{\chi \in \widehat{\mathcal{S}}} (\widehat f^\pm _z (\chi))^k \chi (S) .
\end{align}
for each $S\in \mathcal{S}$.
Evaluating Eq.~\eqref{eq:inversion_formula} at $S=I$ finally yields
\begin{align} \label{eq:enums_solution}
    \sum_{w=0}^{nk} (C_w^{(k,+)} \pm C_w^{(k,-)}) z^w = \frac{1}{2^n} \sum_{i=1}^{N^\pm}  m_i^\pm  \times (\alpha_i^\pm(z))^k  \, .
\end{align}
We can extract $C_w^{(k,+)} + C_w^{(k,-)}$ and $ C_w^{(k,+)} - C_w^{(k,-)}$
from Eq.~\eqref{eq:enums_solution} by taking the coefficient in front of $z^w$.
This allows us to compute $C_w^{(k,+)}$ and $C_w^{(k,-)}$ from the $m^\pm_i $ and $\alpha_i^{\pm}$, which themselves do not depend on $k$.

% \subsection{Partial transpose enumerators of the six-qubit AME state}
% \label{sec:ame_enumerators}

\section{Noise thresholds and   enumerators for an absolutely maximally entangled state}
% \section{Thresholds for absolutely maximally entangled states}
\label{sec:ame_thresholds}

Let us demonstrate our new theoretical framework from Sec.~\ref{sec:new_enumerators} using the example of the $6$-qubit {\emph{absolutely maximally entangled}} (AME) state,
{with the state vector}
\begin{align} \label{eq:ame_definition}
    \ket{\text{AME}} = (\ket{0}\otimes\ket{\bar 0} + \ket{1}\otimes\ket{\bar 1})/\sqrt{2} \, ,
\end{align}
where $\ket{\bar 0}$ and $\ket{\bar1}$ denote the logical basis states of the $\llbracket 5,1,3\rrbracket$ code~\cite{laflamme_perfect_quantum_1996}.
We aim to compute the QWEs, $C^{(k,\theta)}_w$, from Eq.~\eqref{eq:new_enumerator_stab} for all bipartitions $A\vert B$ and moment orders $k\ge2$.
For $k=2$, we already know $C^{(2,\theta)}_w = \delta_{\theta, 1}\delta_{w,2i}  a_i[\Psi_\text{AME}] 2^6$, 
where $\mathbf{a}[\Psi_{\text{AME}}] = (1,0,0,0,45,0,18)/2^6$
is the vector of Shor--Laflamme enumerators
of $\Psi_\text{AME}=\ket{\text{AME}}\!\bra{\text{AME}}$.
These enumerators have recently been measured in a two-copy Bell sampling experiment~\cite{miller_experimental_measurement_2024}, and they play a central role in Gleason's theorem, see App.~\ref{app:gleasons_theorem} for details.

\begin{figure*} 
    \centering
    \includegraphics[width=\textwidth]{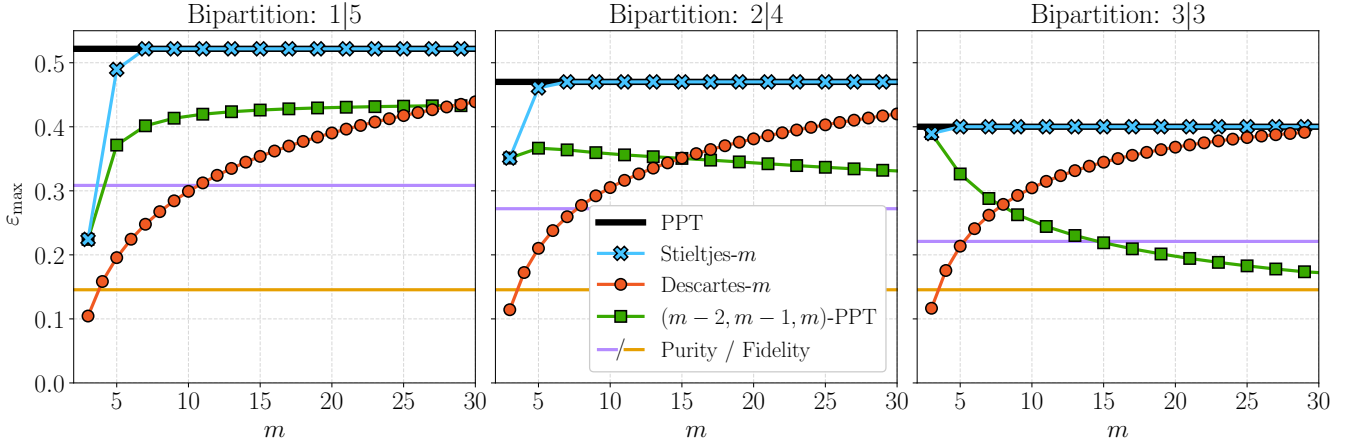}
    \caption{Noise thresholds, $\varepsilon_\text{max}$, below which the locally depolarized six-qubit AME state, $\mathcal{D}_\varepsilon^{\otimes n} [\Psi_\text{AME}]$, is detected by a given entanglement criterion for the three bipartitions of sizes $1\vert5$, $2\vert 4$, and $3\vert 3$. 
    For the Stieltjes-$m$, the Descartes-$m$ and the $(m-2,m-1,m)$-PPT criterion, we show the behavior as a function of the maximum moment order, $m$. 
    All other criteria are shown as constant reference lines.
    The comparison highlights that the Descartes-$m$ criterion performs significantly worse than the corresponding Stieltjes-$m$ criterion.
    The most promising candidates for entanglement detection experiments  are the fidelity criterion (single copies), the purity criterion (two copies), the Stieltjes-$3$, -$5$, and -$7$ criteria.
    For the bipartition $1\vert5$, the $(3,4,5)$-PPT criterion can also be the method of choice, as it is stronger than the Stieltjes-$3$  while requiring less experimental data than the Stieltjes-$5$ criterion.
    }
    \label{fig:ame_loc_depol_thresholds}
\end{figure*}

The essence of Eq.~\eqref{eq:enums_solution} lies in the fact that $C_w^{(k,\theta)}$ 
can be computed efficiently once suitable polynomials and their multiplicities have been identified.
In this section, we slightly modify the notation.
In contrast to Eq.~\eqref{eq:enums_solution}, we write $\alpha_i(z)$ and $\beta_i(z)$ instead of $\alpha_i^-(z)$ and $\alpha_i^+(z)$, and denote their multiplicities by $m_i$ and $n_i$ rather than $m_i^-$ and $m_i^+$.
Since $\hat{f}^-_z(\chi) = \sum_{S\in \mathcal{S}} \chi(S) \vartheta(S) z^{\wt(S)}$ depends on the bipartition $A\vert B$ 
with respect to which the partial transposition is taken (via $\vartheta$),
the same holds for the polynomials $\alpha_i(z)$ and their multiplicities $m_i$.
Due to the symmetry of the AME state, however, they depend only on the size of the bipartition.
Accordingly, we write, e.g., $\alpha_i^{3 \vert 3}(z)$ and $m_i^{3 \vert 3}$ to denote the case of a balanced bipartition such as $A=\{1,2,3\}$ and $B=\{4,5,6\}$.
The polynomials $\beta_i$ and their multiplicities $n_i$, on the other hand, are independent of the specific bipartition under consideration.
A brute-force search reveals that they are given by
\begin{align}
    \beta_1(z) & = 1+ 45z^4+ 18z^6,  \\
    \beta_2(z) &=  1+ 5z^4-6z^6,\, \text{ and}     \\
    \beta_3(z) &=  1 -3z^4+2z^6\, ,
\end{align}
with multiplicities $n_1=1$, $n_2= 18$, and $n_3 = 45$.
The fact that only $3$ out of, in principle, $64$ possible polynomials arise underscores the power of the theoretical framework developed in Sec.~\ref{sec:new_enumerators}.
Note that $\beta_1(z) = \widehat{f}_z^- (\chi_{\text{triv}})$ is nothing other than the Shor--Laflamme weight enumerator polynomial of the AME state, which corresponds to the trivial character $\chi_{\text{triv}} = 1 \in \widehat{\mathcal{S}}$.
Irrespective of the bipartition $A\vert B$, we therefore obtain
\begin{align}
%C^{(k,+)}_w + C^{(k,-)}_w    =  [z^w]  \frac{(1+ 45z^4+ 18z^6)^k +  18(1+ 5z^4-6z^6)^k + 45(1 -3z^4+2z^6)^k}{64}  \, ,
C^{(k,+)}_w &+ C^{(k,-)}_w    =  \frac{[z^w]}{64}  
\biggl(
(1+ 45z^4+ 18z^6)^k \\
&+  18(1+ 5z^4-6z^6)^k + 45(1 -3z^4+2z^6)^k 
\biggr)
\, , \nonumber
\end{align}
where $[z^w]g(z)$ denotes the coefficient of $z^w$ in a polynomial $g(z)$.
The differences, 
\begin{align} \label{eq:ame_differences}
     C^{(k,+)}_w - C^{(k,-)}_w =  [z^w] \frac{1}{2^n}\sum_{i} m_i \alpha_i(z)^k \, ,
\end{align}
on the other hand, depend on the chosen bipartition. 
There are four cases:
\begin{itemize}
    \item[(i)] 
    For $A= \diameter$, we have $\vartheta^{0\vert6}= \chi_\text{triv}$, which implies $\alpha^{0\vert6}_i(z) = \beta(z)$ and $m_i ^{0\vert6}= n_i$ for all $i\in\{1,2,3\}$. 
    In other words, we have $C_w^{(k,-)}=0$, as previously noted in Eq.~\eqref{eq:decay_tr_rho_k}.
    \item[(ii)] 
    For every bipartition of size $1\vert 5$, we find  
    $\alpha_1^{1\vert5}(z) = 6z^6 + 25z^4 + 1$,
    $\alpha_2^{1\vert5}(z) = -2z^6 + z^4 + 1$,
    $\alpha_3^{1\vert5}(z) = 6z^6 - 7z^4 + 1$, and 
    $\alpha_4^{1\vert5}(z) = -18z^6 - 15z^4 + 1$.
    The multiplicities  are given by
    $m_1^{1\vert5} = 3$, $m_2^{1\vert5} = 45$, $m_3^{1\vert5} = 15$, and $ m_4^{1\vert5} = 1$. 
    \item[(iii)] For every bipartition of size $2\vert 4$, we find  
$\alpha_1^{2\vert4}(z) = 2z^6 + 13z^4 + 1$,
$\alpha_2^{2\vert4}(z) = -6z^6 + 5z^4 + 1$,
$\alpha_3^{2\vert4}(z) = 2z^6 - 3z^4 + 1$,
$\alpha_4^{2\vert4}(z) = -6z^6 - 11z^4 + 1$, and 
$\alpha_5^{2\vert4}(z) = 18z^6 - 3z^4 + 1$.
The multiplicities are given by
$m_1^{2\vert4} = 9$, $m_2^{2\vert4} = 12$, $m_3^{2\vert4} = 36$, $m_4^{2\vert4} = 6$, $m_5^{2\vert4} = 1$. 
    \item[(iv)] 
For every bipartition of size $3\vert 3$, we find  
$\alpha_1^{3\vert3}(z) = -2z^6 + 9z^4 + 1$,
$\alpha_2^{3\vert3}(z) = 6z^6 + z^4 + 1$,
$\alpha_3^{3\vert3}(z) = -2z^6 - 7z^4 + 1$, and 
$\alpha_4^{3\vert3}(z) = -18z^6 + 9z^4 + 1$.
The multiplicities are given by
$m_1^{3\vert3} = 18$, $m_2^{3\vert3} = 18$, $m_3^{3\vert3} = 27$, and $m_4^{3\vert3} = 1$.
\end{itemize}

The differences in Eq.~\eqref{eq:ame_differences} are all that is needed to compute the moments $p_k$ of the locally depolarized AME state, recall Eq.~\eqref{eq:pk_local_with_enums_stab}.
The behavior of $p_k$, in turn, is relevant for moment-based relaxations of the PPT entanglement criterion.
We find that the Stieltjes-$m$ criterion (blue crosses in Fig.~\ref{fig:ame_loc_depol_thresholds}) is extremely powerful.
It reproduces the noise thresholds of the abstract PPT criterion,
$\varepsilon^\text{PPT}_\text{max}(1\vert5) \approx 0.52$,
$\varepsilon^\text{PPT}_\text{max}(2\vert4) \approx 0.47$, and
$\varepsilon^\text{PPT}_\text{max}(3\vert3) \approx 0.40$,
already at a maximum moment order of $m = 7$ for the bipartitions $1\vert5$ and $2\vert4$, and even at $m = 5$ for the bipartition $3\vert3$.
These values are shown as black reference lines in Fig.~\ref{fig:ame_loc_depol_thresholds}.
Similarly,  for $m=5$ the Stieltjes thresholds
%, approximately $0.490$ and $0.461$ for $1\vert5$ and $2\vert4$, respectively, 
almost reach the corresponding values of $\varepsilon^\text{PPT}_\text{max}$, while dispensing with the need to measure $p_6$ and $p_7$.
For the Stieltjes-$3$ criterion, also known as the $p_3$-PPT criterion, the thresholds exhibit a more pronounced dependence on the bipartition.
% We find 
% $\varepsilon^{p_3\text{-PPT}}_\text{max}(1\vert5) \approx 0.224$,
% $\varepsilon^{p_3\text{-PPT}}_\text{max}(2\vert4) \approx 0.351$, and
% $\varepsilon^{p_3\text{-PPT}}_\text{max}(3\vert3) \approx 0.389$.
In comparison, the subsystem purity thresholds, which we compute using Shor--Laflamme enumerators~\cite{miller_experimental_measurement_2024},
% are given by 
% $\varepsilon^\text{purity}_\text{max}(1\vert5) \approx 0. $,
% $\varepsilon^\text{purity}_\text{max}(2\vert4) \approx 0. $, and
% $\varepsilon^\text{purity}_\text{max}(3\vert3) \approx 0. $,
are shown by the purple reference lines in Fig.~\ref{fig:ame_loc_depol_thresholds}.
In particular, the subsystem purity criterion is stronger than the $p_3$-PPT criterion for the bipartition $1\vert5$, but not for $2\vert4$ and $3\vert3$.
Finally, the fidelity of $\mathcal{D}^{\otimes n}_\varepsilon[\Psi_\text{AME}]$ drops below $50\,\%$ at $\varepsilon^\text{fidelity}_\text{max} \approx 0.145$ (orange reference line in Fig.~\ref{fig:ame_loc_depol_thresholds}).
In summary, the thresholds generally increase with the experimental difficulty of each criterion, with the exception that $\varepsilon^{p_3\text{-PPT}}_\text{max}(1\vert5) < \varepsilon^\text{purity}_\text{max}(1\vert5)$.

Although, a posteriori, there is not much gain from a practical perspective,
our weight enumerator framework for PT moments enables us to investigate the performance of the Descartes-$m$ criterion and the $(k,l,m)$-PPT criterion from Thm.~\ref{thm:ppt_relaxation_hoelder}.
It has been pointed out that the threshold of the Descartes-$m$ criterion is strictly monotonically increasing in $m$, reproducing the abstract PPT criterion for $m=2^n = 64$~\cite{neven_symmetry_resolved_2021}.
In Fig.~\ref{fig:ame_loc_depol_thresholds}, we observe the predicted monotonic behavior 
(red circles). 
However, the threshold increases only very slowly and does not reach $\varepsilon^\text{PPT}_\text{max}$ even at $m = 30$, whereas the Stieltjes-$m$ criterion already achieves this much earlier.
For completeness, we also plot the thresholds for the $(m-2, m-1, m)$-PPT criterion (green squares).
Interestingly, we observe that the behavior of the curve strongly depends on the bipartition.
For $2\vert4$, the $(m-2,m-1,m)$-PPT threshold grows from $m = 3$ to $m = 5$, before slowly decreasing for larger values of $m$.
For $3\vert3$, on the other hand, the threshold decreases rapidly and monotonically over the entire displayed range.
For $1\vert5$, the green squares show a qualitatively different dependence on $m$: they increase rapidly and plateau around $\varepsilon_\text{max} \approx 0.32$.
The practically most relevant case is $m = 5$, for which we obtain $\varepsilon^{(3,4,5)\text{-PPT}}_\text{max} \approx 0.367$, lying clearly between $\varepsilon^{p_3\text{-PPT}}_\text{max}$ and $\varepsilon^{\text{Stieltjes-}5}_\text{max}$.
For effective per-qubit error rates of $\varepsilon \in (0.22, 0.37)$, 
the $(3,4,5)$-PPT criterion may therefore be preferred in this regime,  as the Stieltjes-$5$ criterion additionally requires measuring the purity $p_2$, whereas the $p_3$-PPT criterion is not strong enough to detect entanglement.

%Overall, these results reveal a rich and bipartition-dependent structure of entanglement  noise thresholds, highlighting that even simple moment-based criteria can capture subtle features of multipartite quantum correlations.

\section{Conclusion} 

In this work, we have taken steps to make entanglement detection
more practical and applicable to real-life experimental platforms. Concretely, we have
%we 
introduced a conceptually simple generalization of the $p_3$-PPT entanglement criterion (Thm.~\ref{thm:ppt_relaxation_hoelder}), which we call the  $(k,l,m)$-PPT criterion.
This approach allows one to compare arbitrary triples of partial transpose (PT) moments, rather than requiring access to all moments up to order $m$, as in previous approaches.
Then, we have generalized previous lower bounds on the logarithmic negativity, which solely rely on a few PT 
moments (Thm.~\ref{thm:quantitative_bounds}).
While completeness of moment-based criteria has been discussed using Descartes-type conditions (based on Newton identities), it has not been established whether the Stieltjes hierarchy (Hankel positivity) alone is complete at finite order. 
Our result (Thm.~\ref{thm:stieltjes_iff_ppt}) shows that for spectra with $N$ distinct eigenvalues, 
the Stieltjes conditions already detect NPT at order $2N-1$.
As important consequences, we showed that for globally depolarized stabilizer states and locally depolarized GHZ states, respectively, the PPT criterion is equivalent to the    Stieltjes-$m$ criterion with $m=5$ (Cor.~\ref{cor:p5_ppt_suffices_for_global_noise}) and $m = 2n + 5$ (Cor.~\ref{cor:stieltjes_for_locally_depolarized_ghz_states}).
Through a detailed analysis of locally depolarized GHZ states, we identified optimal choices of moment orders and showed that our $(3,4,5)$-PPT criterion achieves performance comparable to the more demanding Stieltjes-$5$ criterion. 
For carefully selected triples, we furthermore found that the $(k,l,m)$-criteria retain the favorable noise-scaling behavior of the Stieltjes-$m$ criterion (even if $m$ is constant), outperforming traditional approaches like fidelity- and purity-based methods in relevant regimes.
In the second part of this paper, we introduced a new notion of quantum weight enumerators as a unifying framework to describe the decay of PT moments under local depolarizing noise, providing both analytical insight and computational advantages.
Overall, our results highlight that meaningful relaxations of the PPT criterion can be achieved with fewer experimental resources than previously expected, while maintaining strong detection capabilities. This opens promising avenues for scalable entanglement certification in near-term quantum devices, particularly in regimes where access to certain high-order moments is possible but expensive.

\section{Outlook}

For the examples analyzed in this work, we find that---for a given maximal moment order---the Stieltjes criterion consistently outperforms the Descartes criterion. 
It would be interesting to determine whether this holds in general, either by proving its superiority or by constructing a counterexample.
On the computational side, recent tensor-network methods have been developed to efficiently compute Shor--Laflamme weight enumerator distributions~\cite{cao_quantum_lego_2024, cao_quantum_weight_2024, braccia_computing_exact_2024, pato_planqtn_a_2025}.
A complementary line of work leverages generating-function techniques for their analytical computation~\cite{vallee_sector_length_2026, goodenough_black_white_2026}. 
It is conceivable that some of these approaches could be extended to compute the weight enumerators introduced here 
{\add for certain families of structured quantum states.
For example, we believe that this could be achieved for one-dimensional cluster states, whose second-order enumerators can be computed in various ways~\cite{miller_shor_laflamme_2023, vallee_sector_length_2026}.}
From a theoretical perspective, further investigation of the structure of our new quantum weight enumerators and their relation to other families of quantum invariants appears particularly promising.
A deeper understanding of their algebraic properties may not only enable more efficient evaluation of higher-order moments, but also lead to new and more powerful classes of entanglement criteria. 
Finally, extending our framework to more general noise models could provide additional insight into the robustness of moment-based PPT relaxations and their applicability in large-scale quantum systems.
{\add For example, it would be worthwhile to extend our error analysis from phenomenological to circuit-level noise.
}

% {\red 

% \section*{Further ideas}
% \begin{itemize}
%     \item Man kann ja auch ein SDP einer unvollständigen Hankel-Matrix laufen lassen, um negative Eigenwerte zu zertifizieren.
%     \item Replace Tr[ $\rho W_i$]= $c_i$ by Tr[$(\rho^\Gamma)^k$ ] = $p_k$ in \href{https://arxiv.org/abs/quant-ph/0607167}{Eisert et al}.
% \end{itemize}
% }

\begin{acknowledgments}  
The authors are thankful for stimulating discussions with
Lennart Bittel,
{\add Simon Burton,}
Jose Carrasco,
Philippe Faist, 
Eric J.~Kuehnke,
Kyano Levi,
Arthur Pesah,
and Louis Schatzki.
This work has been supported by the Quantum Flagship (Millenion and PasQuans2), the  BMFTR (RealistiQ, QSolid, MUNIQC-Atoms, DAQC, QuSol, Hybrid++, PasQuops), the Munich Quantum Valley (K-8), the BMWK (EniQmA), the QuantERA (HQCC), the Cluster of Excellence MATH+, the DFG (CRC 183 and SPP 2514:\,563402549), the Einstein Foundation (Einstein Research Unit on Quantum Devices), Berlin Quantum, and the ERC (DebuQC).
This research was sponsored by IARPA and the Army Research Office, under the Entangled Logical Qubits program, and was accomplished under Cooperative Agreement Number W911NF-23-2-0212. The views and conclusions contained in this document are those of the authors and should not be interpreted as representing the official policies, either expressed or implied, of IARPA, the Army Research Office, or the U.S. Government. The U.S. Government is authorized to reproduce and distribute reprints for Government purposes notwithstanding any copyright notation herein. 
\end{acknowledgments}

\appendix

\section{Gleason's theorem about weight enumerators} 
\label{app:gleasons_theorem}

In this appendix, we provide a modern formulation of Gleason's theorem  (1970) from invariant theory~\cite{gleason_weight_polynomials_1970}.
This clarifies the role of the absolutely maximally entangled (AME) state on six qubits,
which is defined in Eq.~\eqref{eq:ame_definition}.
Let $M_{i,j}$ be the coefficient of $x^i$ in the polynomial $(1+3x)^{n-j}(1-x)^j/2^n$.
The self-inverse matrix $M=(M_{i,j})_{i,j=0}^n$ is known as the quantum MacWilliams transformation~\cite{shor_quantum_analog_1997}, and it is diagonalized by $\tilde{T}=( (-1)^{j} M_{i,j})_{i,j=0}^n$~\cite{miller_experimental_measurement_2024}.
Denote by $\ker(M-\mathbbm 1)$ the $(+1)$-eigenspace of $M$, and similarly for $\tilde{T}$.
A pure state $\Psi=\ket{\psi}\!\bra{\psi}$ is said to be a Type-II state
if for all $P\in \{I,X,Y,Z\}^{\otimes n}$ with $\wt(P) =1$ mod $2$, it holds $\bra{\psi} P \ket{\psi} = 0$.
Clearly $\Psi$ is of Type-II iff $a_{2i-1}[\Psi]=0$ for all integer $i$.
Note that Type-II states can only exist if the number of qubits $n$ is even because, by Thm.~2 of Ref.~\cite{tran_quantum_entanglement_2015},
every pure state obeys $a_n[\Psi] \ge 1$.
Finally, we define $\Psi_0 = \ket{0}\!\bra{0}$,  $\Psi_\text{Bell} = \ket{\Phi^+}\!\bra{\Phi^+}$, and  $\Psi_\text{AME}=\ket{\text{AME}}\!\bra{\text{AME}}$,
where $\ket{\Phi^+} = (\ket{0,0}+\ket{1,1})/\sqrt{2}$ is the two-qubit Bell state.

\begin{theorem} \label{thm:gleason} (Gleason's theorem)
Let $\Psi = \ket{\psi}\!\bra{\psi}$ be a pure $n$-qubit state.
Then, its vector of Shor--Laflamme quantum weight enumerators, 
$\mathbf{a}[\Psi] = (a_0[\Psi], \ldots, a_n[\Psi])$, is contained in the $\lceil (n+1)/2\rceil$-dimensional subspace 
\begin{align} \label{eq:gleason1}
    \ker(M-\mathbbm 1) = \mathrm{span}_{\RR} \left\{ \mathbf{a}[\Psi_\mathrm{Bell}^{\otimes i}\otimes \Psi_{0}^{\otimes n-2i}] \ \vert \ 
    % i \in \{0,,\ldots, \left\lfloor \tfrac{n}{2}\right \rfloor\} 
    0 \le i \le \left\lfloor \tfrac{n}{2}\right \rfloor
     \right\} \, .
\end{align}
Moreover, if $\Psi$ is a Type-II state, then $\mathbf{a}[\Psi] $ %$ \in \RR^{n+1}$ 
is contained in the $\lceil (n+1)/6\rceil$-dimensional subspace 
\begin{align}      \label{eq:gleason2}
    \ker &(M-\mathbbm 1) \cap \ker(\tilde{T}-\mathbbm 1)  \\ 
 &=   \mathrm{span}_{\RR} \left \{ \mathbf{a}[\Psi_\mathrm{AME}^{\otimes i}\otimes \Psi_\mathrm{Bell}^{\otimes (n-6i)/2}] \ \vert \      0 \le i \le  \left\lfloor \tfrac{n}{6}\right \rfloor \right\} \, .
    \nonumber
\end{align}
%which has dimension $\lceil (n+1)/6\rceil$.
\end{theorem}
\begin{proof}

{\add 

It is well known that $\mathbf{a}[\Psi]$ is a $(+1)$-eigenvector of the MacWilliams transformation~\cite{huber_some_ulams_2018, miller_experimental_measurement_2024}.
This shows ``$\supseteq$'' in Eq.~\eqref{eq:gleason1}.
For pure stabilizer states, this means that the stabilizer group $\mathcal{S} \subset \mathcal P_n = \{I,X,Y,Z\}^{\otimes n}$ (we ignore global phases) and its dual $\mathcal{S}^\perp = \{P \in \mathcal{P}_n \ |\ \forall S\in \mathcal{S}: SP = PS  \}$
are coinciding groups.
For pure Type-II stabilizer states $\Psi$, such as $\Psi_\text{AME}$ and $\Psi_\text{Bell}$,
the dual $\mathcal{S}^\perp$, in turn, coincides with the shadow $\tilde{\mathcal{S}} = \{ P \in \mathcal{P}_n \ | \ \forall S \in \mathcal{S}: SP = PS (-1)^{\wt(S)}\}$, see e.g., Prop.~2 in Ref.~\cite{miller_experimental_measurement_2024}.
Therefore, the shadow enumerators of $\Psi$ coincide with its dual Shor--Laflamme enumerators, which themselves coincide with $\mathbf{a}[\Psi]$.
This shows $\tilde{T} \mathbf{a}[\Psi] = \mathbf{a}[\Psi] = M\mathbf{a}[\Psi]$ and proves  ``$\supseteq$'' in Eq.~\eqref{eq:gleason2}.

To prove the converse inclusion in Eq.~\eqref{eq:gleason1}, 
let $\mathbf{a} \in \mathbb{R}^{n+1}$ be an arbitrary vector with $M\mathbf{a} = \mathbf{a}$.
Define the polynomial $f(x,y) = \sum_{i=0}^n a_i x^{n-i}y^{i} \in \CC [x,y]$.
Since $\mathbf{a}$ is invariant under the MacWilliams matrix $M$,
the polynomial $f$ is invariant under the variable substitution $(x,y) \mapsto (\frac{x+3y}{2}, \frac{x-y}{2})$~\cite{macwilliams_a_theorem_1962}.
In the language of invariant theory~\cite[Sec.~5.6]{nebe_self_dual_2006}, this means that $f$ is contained in the 
invariant ring $\CC[x,y]^{G_\text{I}}$,
% \begin{align}
%     \CC[x,y]^G = \{f \in \CC [x,y] \ \vert \ f^g = f\}
% \end{align}
where $G_\text{I} = \langle h\rangle \cong C_2$ and
\begin{align} \label{eq:macwilliams_symmetry}
    h = \frac{1}{2} \begin{bmatrix}
        1 & 3 \\ 1 & -1
    \end{bmatrix}  
    \in  \text{GL}_2(\CC)\, .
\end{align}
By \cite[Sec.~7.6.1]{nebe_self_dual_2006},
$G_\text{I}$ arises from the Clifford--Weyl group $\mathcal{C}(4^\text{H+})$~\cite[Def.~5.3.1]{nebe_self_dual_2006}
associated with the form ring representation $\rho(4^\text{H+})$~\cite[Def.~1.7.2]{nebe_self_dual_2006},
which defines %(classical)
Type-$4^\text{H+}$ codes~\cite[Def.~1.8.1]{nebe_self_dual_2006},
after symmetrization to Hamming weight enumerators~\cite[Rem.~after Def.~2.1.2]{nebe_self_dual_2006}.
Hereby, a Type-$4^\text{H+}$ code is an additive subgroup $C\subset \FF_4^n$ that is self-dual with respect to the trace-Hermitian inner product~\cite[Sec.~2.3.4]{nebe_self_dual_2006}.
In other words, the classical code $C$ can be identified with the stabilizer group $\mathcal{S}$ of a pure stabilizer state~\cite[Rem.~13.2.3 (a)]{nebe_self_dual_2006}.
By~\cite[Eqs.~(5.6.8), (7.6.4)]{nebe_self_dual_2006},
we have 
% \begin{align}
    $\CC[x,y]^{G_\text{I}} = \CC [\text{hwe}(i_1),  \text{hwe}(i_2)]$, 
% \end{align}
where $\text{hwe}(i_1) = x+y$ and $\text{hwe}(i_2) = x^2 +3y^2$
are the Hamming weight enumerator polynomials of the Type-$4^\text{H+}$ codes 
$i_1 =\{0,1\}\subset \FF_4$ and $i_2 = \{00,11, \omega\omega,  \bar\omega\bar\omega\}\subset \FF_4^2$~\cite[Eq.(2.4.24)]{nebe_self_dual_2006}.
In the terminology of quantum information theory, the codes $i_1$ and $i_2$ can be identified with the stabilizer groups of $\ket{0}$ and $\ket{\Phi^+}$.
In particular, the Shor--Laflamme weight enumerator polynomials of $\Psi_0$ and $\Psi_\text{Bell}$ form a good polynomial basis~\cite[Def.~5.6.1]{nebe_self_dual_2006}
of $\mathbb{C}[x,y]^{G_\text{I}}$. % = \CC [x+y, x^2 +3y^2]$ is 
Therefore, our previous observation $f \in \CC[x,y] ^{G_\text{I}}$ is equivalent to 
the vector $\mathbf{a}  $ being contained in the right-hand side of Eq.~\eqref{eq:gleason1}.

Finally, let us prove ``$\subseteq$'' in Eq.~\eqref{eq:gleason2}.
This time, let $\mathbf{a}\in \RR^{n+1}$ be a $(+1)$-eigenvector of both $M$ and $\tilde T$.
Again, let $f(x,y) = \sum_{i=0}^n a_i x^{n-i}y^i$ be the polynomial whose coefficients are given by $\mathbf{a}$.
As before, $f$ is invariant under the MacWilliams symmetry $h$ from Eq.~\eqref{eq:macwilliams_symmetry}.
Our assumption $\tilde{T}\mathbf{a} = M\mathbf{a}= \mathbf{a}$ furthermore implies $a_i = 0$ whenever $i$ is odd.
Therefore, the polynomial $f$ is also invariant under the variable transformation $(x,y) \mapsto (x, -y)$, which corresponds to the symmetry 
\begin{align} \label{eq:parity_symmetry_2}
    \phi = \begin{bmatrix}
        1 &0 \\ 0 &-1
    \end{bmatrix}    \in  \text{GL}_2(\CC)\, .
\end{align}
This shows that $f$ is contained in the invariant ring $\CC [x,y] ^{G_\text{II}}$, 
where $G_\text{II} = \langle h,\phi \rangle \cong \text{Sym(3)} \times  C_2$.
By~\cite[Sec.~7.6.3]{nebe_self_dual_2006}, 
$G_\text{II}$ arises from the Clifford--Weyl group $\mathcal{C}(4^\text{H+}_\text{II})$ associated with the form ring representation $\rho(4^\text{H+}_\text{II})$,
which defines Type-$4^\text{H+}_\text{II}$ codes~\cite[p.~50]{nebe_self_dual_2006},
after symmetrization to Hamming weight enumerators.
Hereby, a Type-$4^\text{H+}_\text{II}$ code is a Type-$4^\text{H+}$ code  $C\subset \FF_4^n$ with the additional property that, for every code word $c\in C$, the Hamming weight of $c$ is even~\cite[Rem.~2.3.1]{nebe_self_dual_2006}.
In our terminology, the code $C$ can be identified with the stabilizer group $\mathcal{S}$ of a pure Type-II state.
By~\cite[Eq.~(7.6.23)]{nebe_self_dual_2006}, we have 
$\CC[x,y]^{G_\text{II}} = \CC[\text{hwe}(i_2), \text{hwe}(h_6)]$,
where $\text{hwe}(i_2) = x^2 + 3y^2$ and $\text{hwe}(h_6)= x^6 + 45x^2y^4+ 18y^6$
%, respectively,
are the Hamming weight enumerator polynomials of the repetition code $i_2$~\cite[Eq.~(2.4.24)]{nebe_self_dual_2006} and the hexacode $h_6$~\cite[Eq.~(2.4.26)]{nebe_self_dual_2006}.
These classical codes can be identified with the stabilizer groups of $\ket{\Phi^+}$ and $\ket{\text{AME}}$.
Therefore, the Shor--Laflamme weight enumerator polynomials of $\Psi_\text{Bell}$ and $\Psi_\text{AME}$ form a good polynomial basis of $\CC[x,y]^{G_\text{II}}$, and $f \in \CC[x,y]^{G_\text{II}}$ implies that the vector $\mathbf{a}$ is contained in the right-hand side of Eq.~\eqref{eq:gleason2}, which was to be demonstrated.

}

% The theorem is a direct reformulation of Eqs.~(7.6.4) and~(7.6.23) of Ref.~\cite{nebe_self_dual_2006}. 
\end{proof}

Gleason’s Theorem is useful because it enables the reconstruction of the full vector of Shor--Laflamme enumerators for pure states and pure Type-II states, respectively, even when only $\lceil (n+1)/2\rceil$ and $\lceil (n+1)/6\rceil$ of its entries are known.
To achieve this, it is important to note that $\mathbf{a}[\rho\otimes \rho']$ can be computed from $\mathbf{a}[\rho]$ and $\mathbf{a}[\rho']$, see  Eq.~(21) in Ref.~\cite{wyderka_characterizing_quantum_2020}.
In this way, the basis vectors in Eqs.~\eqref{eq:gleason1} and~\eqref{eq:gleason2} can be directly computed from   $\mathbf{a}[\Psi_{0}] = (1/2,1/2)$,   $\mathbf{a}[\Psi_{\text{Bell}}] = (1/4,0,3/4)$, and $\mathbf{a}[\Psi_{\text{AME}}] = (1/64,0,0,0,45/64,0,18/64)$. 

{\add 
Having translated Gleason’s theorem into the language of modern quantum information theory, 
it would be unfortunate not to mention the corresponding translation of a theorem attributed to Gleason and Pierce~\cite{assmus_research_to_1967}.
This theorem is a no-go result for the existence of pure states $\Psi  $ whose Pauli expectation values $\Tr[\Psi P]$ vanish unless $\wt(P) = 0 \text{ mod }t$, where  $t\ge 3$ is an arbitrary integer.  
This is surprising because such states do exist for $t=2$ (the pure Type-II states).
Moreover, if the purity assumption is dropped, such states exist for all $t$, e.g.,  $\rho = (\mathbbm 1 + Z^{\otimes t})2^{-t}$.

\begin{theorem}(Gleason--Pierce Theorem)
Let $n,t \in \mathbb{N}$ with $n\ge 1$ and $t\ge 3$.
There is no pure $n$-qubit state $\Psi$ whose Shor--Laflamme enumerators obey $a_i[\Psi] = 0$ whenever $i$ is not an integer multiple of $t$.
\end{theorem}
\begin{proof}
    Assume such a state $\Psi$ exists and denote its weight enumerator polynomial by 
    $f(x,y) = \sum_{i=0}^n a_i[\Psi] x^{n-i}y^i$.
    Since every coefficient with $i\neq0 \text{ mod } t$ is vanishing,
    the polynomial $f$ is invariant under the variable transformation $(x,y) \mapsto (1, \omega)$ with $\omega = \exp(2\pi \texttt{i} / t)$.
    This transformation corresponds to the symmetry
    \begin{align}
       \phi_t = \begin{pmatrix}
           1 & 0 \\ 0 & \omega
       \end{pmatrix} \, ,
    \end{align}
    which generalizes Eq.~\eqref{eq:parity_symmetry_2}.
    Since $\Psi$ is pure, $f$ is also invariant under the MacWilliams symmetry of $h$ from Eq.~\eqref{eq:macwilliams_symmetry}.
    Therefore, $f$ is contained in the invariant ring $\mathbb{C}[x,y]^{G_t}$
    of the infinite group $G_t = \langle h, \phi_t\rangle$.
    However, $\mathbb{C}[x,y]^{G_t} = \mathbb C$ is trivial~\cite[Part~XI]{assmus_research_to_1967}.
    Therefore $a_i[\Psi] = \delta_{i,0}2^{-n}$, which contradicts that $\Psi$ is pure and finishes the proof.
\end{proof}

}

\bibliography{references}

@article{arute_quantum_supremacy_2019,
  author  = {Arute, Frank and others},
  title   = {Quantum supremacy using a programmable superconducting processor},
  journal = {Nature},
  year    = {2019},
  volume  = {574},
  number  = {7779},
  pages   = {505--510},
  doi     = {10.1038/s41586-019-1666-5},
  url     = {https://doi.org/10.1038/s41586-019-1666-5},
  issn    = {1476-4687}
}

@article{kim_evidence_for_2023,
  author  = {Kim, Youngseok and others},
  title   = {Evidence for the utility of quantum computing before fault tolerance},
  journal = {Nature},
  year    = {2023},
  volume  = {618},
  number  = {7965},
  pages   = {500--505},
  doi     = {10.1038/s41586-023-06096-3}
}

@article{bluvstein_logical_quantum_2024,
  author  = {Bluvstein, Dolev and others},
  title   = {Logical quantum processor based on reconfigurable atom arrays},
  journal = {Nature},
  year    = {2024},
  volume  = {626},
  number  = {7997},
  pages   = {58--65},
  doi     = {10.1038/s41586-023-06927-3}
}

@article{acharya_quantum_error_2025,
  author  = {Acharya, Rajeev and others},
  title   = {Quantum error correction below the surface code threshold},
  journal = {Nature},
  year    = {2025},
  volume  = {638},
  number  = {8052},
  pages   = {920--926},
  doi     = {10.1038/s41586-024-08449-y}
}

@article{werner_quantum_states_1989,
  title = {{Quantum states with Einstein-Podolsky-Rosen correlations admitting a hidden-variable model}},
  author = {Werner, Reinhard F.},
  journal = {Phys. Rev. A},
  volume = {40},
  issue = {8},
  pages = {4277--4281},
  numpages = {0},
  year = {1989},
  month = {Oct},
  publisher = {American Physical Society},
  doi = {10.1103/PhysRevA.40.4277},
  url = {https://link.aps.org/doi/10.1103/PhysRevA.40.4277}
}

@article{horodecki_violating_bell_1995,
title = {{Violating Bell inequality by mixed spin-1/2 states: necessary and sufficient condition}},
journal = {Phys. Lett. A},
volume = {200},
number = {5},
pages = {340-344},
year = {1995},
issn = {0375-9601},
doi = {https://doi.org/10.1016/0375-9601(95)00214-N},
url = {https://www.sciencedirect.com/science/article/pii/037596019500214N},
author = {R. Horodecki and P. Horodecki and M. Horodecki}
}

@article{simon_robustness_of_2002,
  title = {Robustness of multiparty entanglement},
  author = {Simon, Christoph and Kempe, Julia},
  journal = {Phys. Rev. A},
  volume = {65},
  issue = {5},
  pages = {052327},
  numpages = {4},
  year = {2002},
  month = {May},
  publisher = {American Physical Society},
  doi = {10.1103/PhysRevA.65.052327},
  url = {https://link.aps.org/doi/10.1103/PhysRevA.65.052327}
}

@article{liu_decay_of_2009,
  title = {Decay of multiqudit entanglement},
  author = {Liu, Zhao and Fan, Heng},
  journal = {Phys. Rev. A},
  volume = {79},
  issue = {6},
  pages = {064305},
  numpages = {4},
  year = {2009},
  month = {Jun},
  publisher = {American Physical Society},
  doi = {10.1103/PhysRevA.79.064305},
  url = {https://link.aps.org/doi/10.1103/PhysRevA.79.064305}
}

@book{bhatia_matrix_analysis_1997,
  author    = {Bhatia, Rajendra},
  title     = {Matrix Analysis},
  series    = {Grad. Texts Math.},
  volume    = {169},
  publisher = {Springer},
  address   = {NY},
  year      = {1997},
  edition   = {1},
  doi       = {10.1007/978-1-4612-0653-8},
  isbn      = {978-0-387-94846-1},
  ean       = {9780387948461},
  issn      = {0072-5285},
  eissn     = {2197-5612},
}

@book{terras_fourier_analysis_1999,
    place={Cambridge},  
    title={{Fourier Analysis on Finite Groups and Applications}}, publisher={CUP},
    author={Terras, Audrey},
    year={1999}
}

@article{freedman_experimental_test_1972,
  title = {{Experimental test of local hidden-variable theories}},
  author = {Freedman, Stuart J. and Clauser, John F.},
  journal = {Phys. Rev. Lett.},
  volume = {28},
  issue = {14},
  pages = {938--941},
  numpages = {0},
  year = {1972},
  month = {Apr},
  publisher = {American Physical Society},
  doi = {10.1103/PhysRevLett.28.938},
  url = {https://link.aps.org/doi/10.1103/PhysRevLett.28.938}
}

@article{aspect_experimental_realization_1982,
  title = {{Experimental Realization of Einstein-Podolsky-Rosen-Bohm Gedankenexperiment: A new violation of Bell's inequalities}},
  author = {Aspect, Alain and Grangier, Philippe and Roger, G\'erard},
  journal = {Phys. Rev. Lett.},
  volume = {49},
  issue = {2},
  pages = {91--94},
  numpages = {0},
  year = {1982},
  month = {Jul},
  publisher = {American Physical Society},
  doi = {10.1103/PhysRevLett.49.91},
  url = {https://link.aps.org/doi/10.1103/PhysRevLett.49.91}
}

@article{weihs_violation_of_1998,
  title = {{Violation of Bell's Inequality under strict Einstein locality conditions}},
  author = {Weihs, Gregor and others},
  journal = {Phys. Rev. Lett.},
  volume = {81},
  issue = {23},
  pages = {5039--5043},
  numpages = {0},
  year = {1998},
  month = {Dec},
  publisher = {American Physical Society},
  doi = {10.1103/PhysRevLett.81.5039},
  url = {https://link.aps.org/doi/10.1103/PhysRevLett.81.5039}
}

@article{hensen_loophole_free_2015,
  author  = {Hensen, B. and others},
  title   = {{Loophole-free Bell inequality violation using electron spins separated by 1.3 kilometres}},
  journal = {Nature},
  year    = {2015},
  volume  = {526},
  number  = {7575},
  pages   = {682--686},
  doi     = {10.1038/nature15759},
  url     = {https://doi.org/10.1038/nature15759}
}

@article{storz_loophole_free_2023,
  author  = {Storz, Simon and others},
  title   = {{Loophole-free Bell inequality violation with superconducting circuits}},
  journal = {Nature},
  year    = {2023},
  volume  = {617},
  number  = {7960},
  pages   = {265--270},
  doi     = {10.1038/s41586-023-05885-0},
  url     = {https://doi.org/10.1038/s41586-023-05885-0}
}

@misc{fattal_entanglement_in_2004,
      title={Entanglement in the stabilizer formalism}, 
      author={David Fattal and others},
      year={2004},
      eprint={quant-ph/0406168},
      archivePrefix={arXiv},
}

@article{guehne_entanglement_detection_2009,
title = {Entanglement detection},
journal = {Physics Rep.},
volume = {474},
number = {1},
pages = {1-75},
year = {2009},
issn = {0370-1573},
doi = {https://doi.org/10.1016/j.physrep.2009.02.004},
url = {https://www.sciencedirect.com/science/article/pii/S0370157309000623},
author = {O. Gühne and G. Tóth},
}

@article{bourennane_experimental_detection_2004,
  title = {{Experimental Detection of Multipartite Entanglement using Witness Operators}},
  author = {Bourennane, Mohamed and  others },
  journal = {Phys. Rev. Lett.},
  volume = {92},
  issue = {8},
  pages = {087902},
  numpages = {4},
  year = {2004},
  month = {Feb},
  publisher = {American Physical Society},
  doi = {10.1103/PhysRevLett.92.087902},
  url = {https://link.aps.org/doi/10.1103/PhysRevLett.92.087902}
}

@article{huang_experimental_generation_2011,
  author  = {Huang, Yun-Feng and others},
  title   = {{Experimental generation of an eight-photon Greenberger--Horne--Zeilinger state}},
  journal = {Nat. Commun.},
  year    = {2011},
  volume  = {2},
  number  = {1},
  pages   = {546},
  doi     = {10.1038/ncomms1556},
  url     = {https://doi.org/10.1038/ncomms1556}
}

@article{zhang_experimental_greenberger_2015,
  title = {{Experimental Greenberger-Horne-Zeilinger-type six-photon quantum nonlocality}},
  author = {Zhang, Chao and others},
  journal = {Phys. Rev. Lett.},
  volume = {115},
  issue = {26},
  pages = {260402},
  numpages = {5},
  year = {2015},
  month = {Dec},
  publisher = {American Physical Society},
  doi = {10.1103/PhysRevLett.115.260402},
  url = {https://link.aps.org/doi/10.1103/PhysRevLett.115.260402}
}

@article{wang_18_qubit_2018,
  title = {{18-Qubit entanglement with six photons' three degrees of freedom}},
  author = {Wang, Xi-Lin and others},
  journal = {Phys. Rev. Lett.},
  volume = {120},
  issue = {26},
  pages = {260502},
  numpages = {5},
  year = {2018},
  month = {Jun},
  publisher = {American Physical Society},
  doi = {10.1103/PhysRevLett.120.260502},
  url = {https://link.aps.org/doi/10.1103/PhysRevLett.120.260502}
}

@article{thomas_efficient_generation_2022,
  author  = {Thomas, Philip and Ruscio, Leonardo and Morin, Olivier and Rempe, Gerhard},
  title   = {Efficient generation of entangled multiphoton graph states from a single atom},
  journal = {Nature},
  year    = {2022},
  volume  = {608},
  number  = {7924},
  pages   = {677--681},
  doi     = {10.1038/s41586-022-04987-5},
  url     = {https://doi.org/10.1038/s41586-022-04987-5}
}

@article{mooney_generation_and_2021,
doi = {10.1088/2399-6528/ac1df7},
url = {https://doi.org/10.1088/2399-6528/ac1df7},
year = {2021},
publisher = {IOP Publishing},
volume = {5},
number = {9},
pages = {095004},
author = {Mooney, Gary J. and others},
title = {{Generation and verification of 27-qubit Greenberger-Horne-Zeilinger states in a superconducting quantum computer}},
journal = {J. Phys. Commun.},
}

@article{cao_generation_of_2023,
    author = {S. Cao and others},
    title = {{Generation of genuine entanglement up to 51 superconducting qubits}},
    journal = {Nature},
    year = {2023},
    volume=619, 
    pages={738},
    url = {https://doi.org/10.1038/s41586-023-06195-1},
    doi = {10.1038/s41586-023-06195-1}
}

@article{moses_a_race_2023,
  title = {{A Race-Track Trapped-Ion Quantum Processor}},
  author = {Moses, S. A. and others},
  journal = {Phys. Rev. X},
  volume = {13},
  issue = {4},
  pages = {041052},
  numpages = {25},
  year = {2023},
  month = {Dec},
  publisher = {American Physical Society},
  doi = {10.1103/PhysRevX.13.041052},
  url = {https://link.aps.org/doi/10.1103/PhysRevX.13.041052}
}

@article{kam_characterization_of_2024,
  title = {Characterization of entanglement on superconducting quantum computers of up to 414 qubits},
  author = {Kam, John F. and others},
  journal = {Phys. Rev. Res.},
  volume = {6},
  issue = {3},
  pages = {033155},
  numpages = {23},
  year = {2024},
  month = {Aug},
  publisher = {American Physical Society},
  doi = {10.1103/PhysRevResearch.6.033155},
  url = {https://link.aps.org/doi/10.1103/PhysRevResearch.6.033155}
}

@misc{javadiabhari_big_cats_2025,
      title={Big cats: entanglement in 120 qubits and beyond}, 
      author={Ali Javadi-Abhari and others},
      year={2025},
      eprint={2510.09520},
      archivePrefix={arXiv},
}

@article{leibfried_creation_of_2005,
  author  = {Leibfried, D. and others},
  title   = {Creation of a six-atom {Schr{\"o}dinger} cat state},
  journal = {Nature},
  year    = {2005},
  volume  = {438},
  number  = {7068},
  pages   = {639--642},
  doi     = {10.1038/nature04251},
  url     = {https://doi.org/10.1038/nature04251}
}

@article{omran_generation_of_2019,
author = {A. Omran  and  others },
title = {{Generation and manipulation of Schrödinger cat states in Rydberg atom arrays}},
journal = {Science},
volume = {365},
number = {6453},
pages = {570-574},
year = {2019},
doi = {10.1126/science.aax9743},
URL = {https://www.science.org/doi/abs/10.1126/science.aax9743},
}

@article{peres_separability_criterion_1996,
  title = {{Separability criterion for density matrices}},
  author = {Peres, A.},
  journal = {Phys. Rev. Lett.},
  volume = {77},
  issue = {8},
  pages = {1413--1415},
  numpages = {0},
  year = {1996},
  publisher = {American Physical Society},
  doi = {10.1103/PhysRevLett.77.1413},
  url = {https://link.aps.org/doi/10.1103/PhysRevLett.77.1413}
}

@article{horodecki_separability_1996,
title = {{Separability of mixed states: necessary and sufficient conditions}},
journal = {Phys. Lett. A},
volume = {223},
number = {1},
pages = {1-8},
year = {1996},
issn = {0375-9601},
doi = {https://doi.org/10.1016/S0375-9601(96)00706-2},
url = {https://www.sciencedirect.com/science/article/pii/S0375960196007062},
author = {M. Horodecki and P. Horodecki and R. Horodecki},
}

@article{plenio_logarithmic_negativity_2005,
  title = {{Logarithmic Negativity: A Full entanglement monotone that is not convex}},
  author = {Plenio, M. B.},
  journal = {Phys. Rev. Lett.},
  volume = {95},
  issue = {9},
  pages = {090503},
  numpages = {4},
  year = {2005},
  month = {Aug},
  publisher = {American Physical Society},
  doi = {10.1103/PhysRevLett.95.090503},
  url = {https://link.aps.org/doi/10.1103/PhysRevLett.95.090503}
}

@article{macwilliams_a_theorem_1962,
  author    = {{\add F.~J.~MacWilliams}},
  title     = {{\add A theorem on the distribution of weights in a systematic code}},
  journal   = {The Bell System Tech. J.},
  volume    = {42},
  pages     = {79--94},
  year      = {1963},
  publisher = {Bell Telephone Laboratories, Inc.},
  doi       = {10.1002/j.1538-7305.1963.tb04003.x}
}

@article{shor_quantum_analog_1997,
  title = {{Quantum Analog of the MacWilliams Identities for Classical Coding Theory}},
  author = {Shor, Peter and Laflamme, Raymond},
  journal = {Phys. Rev. Lett.},
  volume = {78},
  issue = {8},
  pages = {1600--1602},
  numpages = {0},
  year = {1997},
  month = {Feb},
  publisher = {American Physical Society},
  doi = {10.1103/PhysRevLett.78.1600},
  url = {https://link.aps.org/doi/10.1103/PhysRevLett.78.1600}
}

@article{rains_quantum_weight_1998,
  author={Rains, E.M.},
  journal={IEEE Trans. Inf. Theory}, 
  title={Quantum weight enumerators}, 
  year={1998},
  volume={44},
  number={4},
  pages={1388-1394},
  doi={10.1109/18.681316}}

@article{rains_polynomial_invariants_2000,
  author={Rains, E.M.},
  journal = {IEEE Trans. Inf. Theory},
  title={Polynomial invariants of quantum codes}, 
  year={2000},
  volume={46},
  number={1},
  pages={54-59},
  doi={10.1109/18.817508}}

@article{grassl_computing_local_1998,
  title = {{\add Computing local invariants of quantum-bit systems}},
  author = {{\add M.~Grassl, M.~R\"otteler, and T.~Beth}},
  journal = {Phys. Rev. A},
  volume = {58},
  issue = {3},
  pages = {1833--1839},
  numpages = {0},
  year = {1998},
  month = {Sep},
  publisher = {American Physical Society},
  doi = {10.1103/PhysRevA.58.1833},
  url = {https://link.aps.org/doi/10.1103/PhysRevA.58.1833}
}

@article{rains_a_semidefinite_2001,
  author={Rains, E.M.},
  journal = {IEEE Trans. Inf. Theory},
  title={A semidefinite program for distillable entanglement}, 
  year={2001},
  volume={47},
  number={7},
  pages={2921-2933},
  doi={10.1109/18.959270}}

@article{aschauer_local_invariants_2004,
  author  = {H. Aschauer and J. Calsamiglia and M. Hein and H. J. Briegel},
  title   = {Local invariants for multi-partite entangled states allowing for a simple entanglement criterion},
  journal = {Quant. Inf. Comput. 4, 383},
  year    = {2004},
  url     = {https://arxiv.org/abs/quant-ph/0306048}
}

@article{cao_quantum_lego_2024,
  title = {{Quantum Lego expansion pack: Enumerators from tensor networks}},
  author = {Cao, C.-J. and Gullans, M. J. and Lackey, B. and Wang, Z.},
  journal = {PRX Quantum},
  volume = {5},
  issue = {3},
  pages = {030313},
  numpages = {43},
  year = {2024},
  month = {Jul},
  publisher = {American Physical Society},
  doi = {10.1103/PRXQuantum.5.030313},
  opturl = {https://link.aps.org/doi/10.1103/PRXQuantum.5.030313}
}

@ARTICLE{cao_quantum_weight_2024,
  author={Cao, C.-J. and Lackey, B.},
  journal={IEEE Trans. Inf. Theory}, 
  title={Quantum weight enumerators and tensor networks}, 
  year={2024},
  volume={70},
  number={5},
  pages={3512-3528},
  doi={10.1109/TIT.2023.3340503}}

@article{braccia_computing_exact_2024,
  author    = {P. Braccia and P. Bermejo and L. Cincio and M. Cerezo},
  title     = {Computing exact moments of local random quantum circuits via tensor networks},
  journal   = {Quantum Mach. Intell.},
  volume    = {6},
  number    = {2},
  pages     = {54},
  year      = {2024},
  issn      = {2524-4914},
  doi       = {10.1007/s42484-024-00187-8},
  url       = {https://doi.org/10.1007/s42484-024-00187-8}
}

@software{pato_planqtn_a_2025,
 author = {Pato, B. and Vanlerberghe, J. and Cao, C. and Lackey, B. and Brown, K.},
 title = {{PlanqTN, a Python library and interactive web app implementing the quantum LEGO framework}},
 year = 2025,
 publisher = {Zenodo},
 doi = {10.5281/zenodo.16761072},
 url = {https://doi.org/10.5281/zenodo.16761072},
}

@article{terhal_detecting_quantum_2002,
title={Detecting quantum entanglement},
author={B.M. Terhal},
journal={J. Theo. Comp. Sc.},
volume= {287},
pages={313-335},
year= {2002},
doi={10.48550/arXiv.quant-ph/0101032}
}

@article{guehne_detection_of_2002,
  title = {Detection of entanglement with few local measurements},
  author = {G\"uhne, O. and  others},
  journal = {Phys. Rev. A},
  volume = {66},
  issue = {6},
  pages = {062305},
  numpages = {5},
  year = {2002},
  month = {Dec},
  publisher = {American Physical Society},
  doi = {10.1103/PhysRevA.66.062305},
  url = {https://link.aps.org/doi/10.1103/PhysRevA.66.062305}
}

@INPROCEEDINGS{chen_exponential_separations_2022,
  author={Chen, S. and Cotler, J. and Huang, H.-Y. and Li, J.},
  booktitle={Proc. IEEE 62nd Annu. Symp. Found. Comput. Sci. (FOCS 2021)},
  title={{Exponential Separations Between Learning With and Without Quantum Memory}}, 
  year={2022},
  volume={},
  number={},
  pages={574},
  doi={10.1109/FOCS52979.2021.00063}
}

@article{eisert_quantum_certification_2020,
    title={Quantum certification and benchmarking},
    doi={10.1038/s42254-020-0186-4},
    volume= {2},
    issue = {7},
    pages={382},
    Author={J. Eisert and D. Hangleiter and N. Walk and I. Roth and D. Markham and R. Parekh and U. Chabaud and E. Kashefi},
    journal={Nature Rev. Phys.},
    year = {2020},
    url = {https://doi.org/10.1038/s42254-020-0186-4},
}

@article{elben_mixed_state_2020,
  title = {{Mixed-state entanglement from local randomized measurements}},
  author = {A. Elben and others},
  journal = {Phys. Rev. Lett.},
  volume = {125},
  issue = {20},
  pages = {200501},
  numpages = {6},
  year = {2020},
  publisher = {American Physical Society},
  doi = {10.1103/PhysRevLett.125.200501},
  url = {https://link.aps.org/doi/10.1103/PhysRevLett.125.200501}
}

@article{zhou_single_copies_2020,
  title = {{Single-Copies Estimation of Entanglement Negativity}},
  author = {Zhou, You and Zeng, Pei and Liu, Zhenhuan},
  journal = {Phys. Rev. Lett.},
  volume = {125},
  issue = {20},
  pages = {200502},
  numpages = {6},
  year = {2020},
  month = {Nov},
  publisher = {American Physical Society},
  doi = {10.1103/PhysRevLett.125.200502},
  url = {https://link.aps.org/doi/10.1103/PhysRevLett.125.200502}
}

@article{yu_optimal_entanglement_2021,
  title = {{Optimal entanglement certification from moments of the partial transpose}},
  author = {Yu, X.-D. and Imai, S. and G\"uhne, O.},
  journal = {Phys. Rev. Lett.},
  volume = {127},
  issue = {6},
  pages = {060504},
  numpages = {6},
  year = {2021},
  publisher = {American Physical Society},
  doi = {10.1103/PhysRevLett.127.060504},
  url = {https://link.aps.org/doi/10.1103/PhysRevLett.127.060504}
}

@article{liu_detecting_entanglement_2022,
  title = {{Detecting Entanglement in Quantum Many-Body Systems via Permutation Moments}},
  author = {Liu, Zhenhuan and others},
  journal = {Phys. Rev. Lett.},
  volume = {129},
  issue = {26},
  pages = {260501},
  numpages = {6},
  year = {2022},
  month = {Dec},
  publisher = {American Physical Society},
  doi = {10.1103/PhysRevLett.129.260501},
  url = {https://link.aps.org/doi/10.1103/PhysRevLett.129.260501}
}

@article{huber_some_ulams_2018,
doi = {10.1088/1751-8121/aadd1e},
url = {https://doi.org/10.1088/1751-8121/aadd1e},
year = {2018},
month = {sep},
publisher = {IOP Publishing},
volume = {51},
number = {43},
pages = {435301},
author = {{\add F.~Huber and S.~Severini}},
title = {{\add Some Ulam’s reconstruction problems for quantum states}},
journal = {J. Phys. A},
}

@article{carteret_noiseless_quantum_2005,
  title = {{\add Noiseless Quantum Circuits for the Peres Separability Criterion}},
  author = {{\add H.~A.~Carteret}},
  journal = {Phys. Rev. Lett.},
  volume = {94},
  issue = {4},
  pages = {040502},
  numpages = {4},
  year = {2005},
  month = {Jan},
  publisher = {American Physical Society},
  doi = {10.1103/PhysRevLett.94.040502},
  url = {https://link.aps.org/doi/10.1103/PhysRevLett.94.040502}
}

@article{carrasco_entanglement_phase_2024,
  title = {Entanglement phase diagrams from partial transpose moments},
  author = {Carrasco, Jose and others},
  journal = {Phys. Rev. A},
  volume = {109},
  issue = {1},
  pages = {012422},
  numpages = {19},
  year = {2024},
  month = {Jan},
  publisher = {American Physical Society},
  doi = {10.1103/PhysRevA.109.012422},
  url = {https://link.aps.org/doi/10.1103/PhysRevA.109.012422}
}

@misc{tarabunga_quantifying_mixed_2025,
      title={Quantifying mixed-state entanglement via partial transpose and realignment moments}, 
      author={Poetri Sonya Tarabunga and Tobias Haug},
      year={2025},
      eprint={2507.13840},
      archivePrefix={arXiv},
}

@misc{bradshaw_a_closed_2025,
      title={{A closed form for moment-based entanglement tests associated to the PPT criterion}}, 
      author={Zachary P. Bradshaw and Margarite L. LaBorde},
      year={2025},
      eprint={2503.17525},
      archivePrefix={arXiv},
}

@article{carvalho_decoherence_and_2004,
  title = {{Decoherence and multipartite entanglement}},
  author = {Carvalho, Andr\'e R. R. and Mintert, Florian and Buchleitner, Andreas},
  journal = {Phys. Rev. Lett.},
  volume = {93},
  issue = {23},
  pages = {230501},
  numpages = {4},
  year = {2004},
  publisher = {American Physical Society},
  doi = {10.1103/PhysRevLett.93.230501},
  url = {https://link.aps.org/doi/10.1103/PhysRevLett.93.230501}
}

@article{mintert_concurrence_of_2005,
    title = {{Concurrence of mixed multipartite quantum states}},
    author = {Mintert, F. and Ku\ifmmode \acute{s}\else \'{s}\fi{}, M. and Buchleitner, A.},
    journal = {Phys. Rev. Lett.},
    volume = {95},
    issue = {26},
    pages = {260502},
    numpages = {4},
    year = {2005},
    publisher = {American Physical Society},
    doi = {10.1103/PhysRevLett.95.260502},
    url = {https://link.aps.org/doi/10.1103/PhysRevLett.95.260502}
}

@article{mintert_observable_entanglement_2007,
  title = {{Observable entanglement measure for mixed quantum states}},
  author = {Mintert, F. and Buchleitner, A.},
  journal = {Phys. Rev. Lett.},
  volume = {98},
  issue = {14},
  pages = {140505},
  numpages = {3},
  year = {2007},
  publisher = {American Physical Society},
  doi = {10.1103/PhysRevLett.98.140505},
  url = {https://link.aps.org/doi/10.1103/PhysRevLett.98.140505}
}

@article{aolita_scalable_method_2008,
  title = {{Scalable method to estimate experimentally the entanglement of multipartite systems}},
  author = {Aolita, L. and Buchleitner, A. and Mintert, F.},
  journal = {Phys. Rev. A},
  volume = {78},
  issue = {2},
  pages = {022308},
  numpages = {4},
  year = {2008},
  publisher = {American Physical Society},
  doi = {10.1103/PhysRevA.78.022308},
  url = {https://link.aps.org/doi/10.1103/PhysRevA.78.022308}
}

@article{neven_symmetry_resolved_2021,
    author = {A. Neven and others},
    title = {{Symmetry-resolved entanglement detection using partial transpose moments}},
    year = {2021},
    journal = {npj Quant. Inf.},
    pages = {152},
    volume = {7},
    issue = {1},
    doi = {10.1038/s41534-021-00487-y}
}

@article{flammia_direct_fidelity_2011,
  title                    = {{Direct fidelity estimation from few Pauli measurements}},
  Author                   = {Flammia, S.T. and Liu, Y.-K.},
  Journal                  = {Phys. Rev. Lett.},
  Year                     = {2011},
  Pages                    = {230501},
  Volume                   = {106},
  Doi                      = {10.1103/PhysRevLett.106.230501},
  Issue                    = {23},
  Numpages                 = {4}
}

@article{islam_measuring_entanglement_2015,
    author = {R. Islam and others},
    year = {2015},
    title = {{Measuring entanglement entropy in a quantum many-body system}},
    journal = {Nature},
    pages = {77},
    volume = {528},
    issue = {7580},
    doi = {10.1038/nature15750}
}

@article{bovino_direct_measurement_2005,
  title = {{Direct measurement of nonlinear properties of bipartite quantum states}},
  author = {Bovino, Fabio Antonio and others},
  journal = {Phys. Rev. Lett.},
  volume = {95},
  issue = {24},
  pages = {240407},
  numpages = {4},
  year = {2005},
  month = {Dec},
  publisher = {American Physical Society},
  doi = {10.1103/PhysRevLett.95.240407},
  url = {https://link.aps.org/doi/10.1103/PhysRevLett.95.240407}
}

@article{kaufmann_quantum_thermalization_2016,
author = {Adam M. Kaufman  and others},
title = {Quantum thermalization through entanglement in an isolated many-body system},
journal = {Science},
volume = {353},
number = {6301},
pages = {794-800},
year = {2016},
doi = {10.1126/science.aaf6725},
URL = {https://www.science.org/doi/abs/10.1126/science.aaf6725},
}

@article{bluvstein_a_quantum_2022,
  author  = {Bluvstein, Dolev and others},
  title   = {A quantum processor based on coherent transport of entangled atom arrays},
  journal = {Nature},
  year    = {2022},
  volume  = {604},
  number  = {7906},
  pages   = {451--456},
  doi     = {10.1038/s41586-022-04592-6},
  url     = {https://doi.org/10.1038/s41586-022-04592-6}
}

@article{tran_quantum_entanglement_2015,
  title = {Quantum entanglement from random measurements},
  author = {Tran, Minh Cong and others},
  journal = {Phys. Rev. A},
  volume = {92},
  issue = {5},
  pages = {050301},
  numpages = {7},
  year = {2015},
  month = {Nov},
  publisher = {American Physical Society},
  doi = {10.1103/PhysRevA.92.050301},
  url = {https://link.aps.org/doi/10.1103/PhysRevA.92.050301}
}

@misc{tanizawa_simplest_fidelity_2023,
      title={Simplest fidelity-estimation method for graph states with depolarizing noise}, 
      author={Tomonori Tanizawa and others},
      year={2023},
      eprint={2304.10952},
      archivePrefix={arXiv},
}

@inproceedings{gleason_weight_polynomials_1970,
  author    = {Andrew Mattei Gleason},
  title     = {{Weight polynomials of self-dual codes and the MacWilliams identities}},
  booktitle = {Actes Congr. Internat. Math.},
  volume    = {3},
  pages     = {211--215},
  year      = {1970},
}

@book{nebe_self_dual_2006,
  title     = {{Self-Dual Codes and Invariant Theory}},
  author    = {Gabriele Nebe and Eric M. Rains and N. J. A. Sloane},
  series    = {Algorithms and Computation in Mathematics},
  volume    = {17},
  publisher = {Springer},
  year      = {2006},
  isbn      = {978-3-540-28451-5},
  doi       = {10.1007/3-540-30731-1},
  url       = {https://doi.org/10.1007/3-540-30731-1}
}

@article{laflamme_perfect_quantum_1996,
  title = {{Perfect quantum error correcting code}},
  author = {Laflamme, Raymond and Miquel, Cesar and Paz, Juan Pablo and Zurek, Wojciech Hubert},
  journal = {Phys. Rev. Lett.},
  volume = {77},
  issue = {1},
  pages = {198--201},
  numpages = {0},
  year = {1996},
  month = {Jul},
  publisher = {American Physical Society},
  doi = {10.1103/PhysRevLett.77.198},
  url = {https://link.aps.org/doi/10.1103/PhysRevLett.77.198}
}

@misc{miller_experimental_measurement_2024,
      title={Experimental measurement and a physical interpretation of quantum shadow enumerators}, 
  author = {{D. Miller et al.}},
      year={2024},
      eprint={2408.16914},
      archivePrefix={arXiv},
      primaryClass={quant-ph},
      url={https://arxiv.org/abs/2408.16914}, 
}

@article{wyderka_characterizing_quantum_2020,
doi = {10.1088/1751-8121/ab7f0a},
url = {https://doi.org/10.1088/1751-8121/ab7f0a},
year = {2020},
month = {jul},
publisher = {IOP Publishing},
volume = {53},
number = {34},
pages = {345302},
author = {Wyderka, N and Gühne, O},
title = {Characterizing quantum states via sector lengths},
journal = {J. Phys. A},
}

@article{miller_shor_laflamme_2023,
doi = {10.1088/1751-8121/ace8d4},
url = {https://doi.org/10.1088/1751-8121/ace8d4},
year = {2023},
publisher = {IOP Publishing},
volume = {56},
number = {33},
pages = {335303},
author = {Miller, Daniel and others},
title = {{Shor–Laflamme distributions of graph states and noise robustness of entanglement}},
journal = {J. Phys. A},
}

@article{kay_optimal_detection_2011,
  title = {{Optimal detection of entanglement in Greenberger-Horne-Zeilinger states}},
  author = {Kay, Alastair},
  journal = {Phys. Rev. A},
  volume = {83},
  issue = {2},
  pages = {020303},
  numpages = {4},
  year = {2011},
  month = {Feb},
  publisher = {American Physical Society},
  doi = {10.1103/PhysRevA.83.020303},
  url = {https://link.aps.org/doi/10.1103/PhysRevA.83.020303}
}

@article{horodecki_method_for_2002,
  title = {{Method for direct detection of quantum entanglement}},
  author = {Horodecki, P. and Ekert, A.},
  journal = {Phys. Rev. Lett.},
  volume = {89},
  issue = {12},
  pages = {127902},
  numpages = {4},
  year = {2002},
  publisher = {American Physical Society},
  doi = {10.1103/PhysRevLett.89.127902},
  url = {https://link.aps.org/doi/10.1103/PhysRevLett.89.127902}
}

@article{ekert_direct_estimation_2002,
  title = {{Direct estimations of linear and nonlinear functionals of a quantum state}},
  author = {Ekert, Artur K. and others},
  journal = {Phys. Rev. Lett.},
  volume = {88},
  issue = {21},
  pages = {217901},
  numpages = {4},
  year = {2002},
  month = {May},
  publisher = {American Physical Society},
  doi = {10.1103/PhysRevLett.88.217901},
  url = {https://link.aps.org/doi/10.1103/PhysRevLett.88.217901}
}

@article{leifer_measuring_polynomial_2004,
  title = {{Measuring polynomial invariants of multiparty quantum states}},
  author = {Leifer, M. S. and Linden, N. and Winter, A.},
  journal = {Phys. Rev. A},
  volume = {69},
  issue = {5},
  pages = {052304},
  numpages = {8},
  year = {2004},
  month = {May},
  publisher = {American Physical Society},
  doi = {10.1103/PhysRevA.69.052304},
  url = {https://link.aps.org/doi/10.1103/PhysRevA.69.052304}
}

@article{subasi_entanglement_spectroscopy_2019,
doi = {10.1088/1751-8121/aaf54d},
url = {https://doi.org/10.1088/1751-8121/aaf54d},
year = {2019},
month = {jan},
publisher = {IOP Publishing},
volume = {52},
number = {4},
pages = {044001},
author = {Subaşı, Yiğit and Cincio, Lukasz and Coles, Patrick J},
title = {Entanglement spectroscopy with a depth-two quantum circuit},
journal = {J. Phys. A},
}

@article{oszmaniec_measuring_relational_2024,
doi = {10.1088/1367-2630/ad1a27},
url = {https://doi.org/10.1088/1367-2630/ad1a27},
year = {2024},
month = {jan},
publisher = {IOP Publishing},
volume = {26},
number = {1},
pages = {013053},
author = {Oszmaniec, Michał and Brod, Daniel J and Galvão, Ernesto F},
title = {{Measuring relational information between quantum states, and applications}},
journal = {New J. Phys.},
}

@article{quek_multivariate_trace_2024,
  doi = {10.22331/q-2024-01-10-1220},
  url = {https://doi.org/10.22331/q-2024-01-10-1220},
  title = {{Multivariate trace estimation in constant quantum depth}},
  author = {Quek, Yihui and Kaur, Eneet and Wilde, Mark M.},
  journal = {{Quantum}},
  issn = {2521-327X},
  publisher = {{Verein zur F{\"{o}}rderung des Open Access Publizierens in den Quantenwissenschaften}},
  volume = {8},
  pages = {1220},
  month = jan,
  year = {2024}
}

@misc{vallee_sector_length_2026,
      title={{Sector length distributions of recursively definable graph states through analytic combinatorics}}, 
      author={Eloïc Vallée and Kenneth Goodenough and Paul E. Gunnells and Tim Coopmans and Jordi Tura},
      year={2026},
      eprint={2604.09766},
      archivePrefix={arXiv},
      primaryClass={quant-ph},
      url={https://arxiv.org/abs/2604.09766}, 
}

@misc{goodenough_black_white_2026,
      title={{\add Black-white polynomials of graphs and generating functions}}, 
      author={{\add K.~Goodenough and P.~E.~Gunnells}},
      year={2026},
      eprint={2604.10719},
      archivePrefix={arXiv},
      primaryClass={math.CO},
      url={https://arxiv.org/abs/2604.10719}, 
}

@article{ali_partial_transpose_2023,
  author  = {Ali, Mazhar},
  title   = {Partial transpose moments, principal minors and entanglement detection},
  journal = {Quantum Inf. Process.},
  year    = {2023},
  volume  = {22},
  number  = {5},
  pages   = {207},
  doi     = {10.1007/s11128-023-03966-7},
  url     = {https://doi.org/10.1007/s11128-023-03966-7},
  issn    = {1573-1332},
  date    = {2023-05-10}
}

@article{chen_a_matrix_2003,
  author  = {{\add K.~Chen and L.~Wu}},
  title   = {{\add A Matrix Realignment Method for Recognizing Entanglement}},
  journal = {Quantum Inf. Comput.},
  volume  = {3},
  number  = {3},
  pages   = {193},
  year    = {2003},
  doi = {10.48550/arXiv.quant-ph/0205017}
}

@article{rudolph_further_results_2005,
  author  = {{\add O.~Rudolph}},
  title   = {{\add Further Results on the Cross Norm Criterion for Separability}},
  journal = {Quantum Inf. Process.},
  volume  = {4},
  number  = {3},
  pages   = {219},
  year    = {2005},
  doi     = {10.1007/s11128-005-5664-1}
}

@article{horodecki_separability_of_2006,
  author  = {{\add M.~Horodecki, P.~Horodecki, and R.~Horodecki}},
  title   = {{\add Separability of Mixed Quantum States: Linear Contractions and Permutation Criteria}},
  journal = {Open Syst. Inf. Dyn.},
  year    = {2006},
  volume   = {13},
  number   = {1},
  pages    = {103--111},
  doi      = {10.1007/s11080-006-7271-8},
  url      = {https://doi.org/10.1007/s11080-006-7271-8},
  issn     = {1573-1324}
}

@techreport{assmus_research_to_1967,
  title={{\add Research to develop the algebraic theory of codes}},
  author={{\add E.~F.~Assmus, H.~F.~Mattson, and R.~J.~Turyn}},
  year={1967},
  institution = {\add Applied Research Laboratory, Sylvania Electronic Systems}
}

@INPROCEEDINGS{miller_graphstatevis_interactive_2021,
  author={{\add M.~Miller  and D.~Miller}},
  booktitle={Proc. IEEE Int. Conf. Quantum Comput. Eng.}, 
  title={{\add GraphStateVis: Interactive Visual Analysis of Qubit Graph States and their Stabilizer Groups}}, 
  year={2021},
  pages={378},
  doi={10.1109/QCE52317.2021.00057}}

@article{Roads, 
     title ={{Roads towards fault-tolerant universal quantum computation}},
author={{E.~T.~Campbell, B.~M.~Terhal, and C.~Vuillot}},
doi={10.1038/nature23460},
journal={Nature}, volume=549, pages={172-179}, year=2017}

@misc{MindTheGaps,
      title={Mind the gaps: The fraught road to quantum advantage},
      author={{J.~Eisert and J.~Preskill}}, 
      year={2025},
      eprint={2510.19928},
      archivePrefix={arXiv}
}

@article{PRXQuantum.2.010201,
  title = {Theory of Quantum System Certification},
  author = {{M.~Kliesch and I.~Roth}},
  journal = {PRX Quantum},
  volume = {2},
  issue = {1},
  pages = {010201},
  numpages = {53},
  year = {2021},
  month = {Jan},
  publisher = {American Physical Society},
  doi = {10.1103/PRXQuantum.2.010201},
  url = {https://link.aps.org/doi/10.1103/PRXQuantum.2.010201}
}

@article{PhysRevLett.134.210201,
  title = {Experimental Realization of Genuine Three-Copy Collective Measurements for Optimal Information Extraction},
  author = {Zhou, Kai and Yi, Changhao and Yan, Wen-Zhe and Hou, Zhibo and Zhu, Huangjun and Xiang, Guo-Yong and Li, Chuan-Feng and Guo, Guang-Can},
  journal = {Phys. Rev. Lett.},
  volume = {134},
  pages = {210201},
  numpages = {6},
  year = {2025},
  publisher = {American Physical Society},
  doi = {10.1103/PhysRevLett.134.210201} 
}

@article{Hangleiter,
  title = {{Direct certification of a class of quantum simulations}},
  doi={10.1088/2058-9565/2/1/015004},
  Author = {{D.~Hangleiter, M.~Kliesch, M.~Schwarz, and J.~Eisert}},
   journal={Quantum Sci. Technol.},
   volume= 2, 
   pages=015004 , 
   year=2017
}

@misc{PhD,
  Author                   = {J. Eisert},
  eprint={quant-ph/0610253},
      archivePrefix={arXiv},
    title ={Entanglement in quantum information theory},
  Year                     = {2001},
  note                    = {{PhD thesis, University of Potsdam}}
}

@Article{VidalNegativity,
  Author                   = {G. Vidal and R. F. Werner},
  Journal                  = {Phys. Rev. A},
    title ={Computable measure of entanglement},
  Year                     = {2002},
    DOI={10.1103/PhysRevA.65.032314},
  Pages                    = {032314},
  Volume                   = {65}
}

@Article{Audenaert06,
  title                     = {When are correlations quantum?},
  Author                   = {K. M. R. Audenaert and M. B. Plenio},
  Journal                  = {New J. Phys.},
  Year                     = {2006},
    doi={10.1088/1367-2630/8/11/266},
  Pages                    = {266},
  Volume                   = {8}
}

@article{Volume,
author={K.~Zyczkowski and P.~Horodecki and A.~Sanpera and M.~Lewenstein},
  title ={Volume of the set of separable states},
    doi = {10.1103/PhysRevA.58.883},
journal={Phys. Rev. A}, 
volume=58, 
pages=883, 
year=1998}

@Article{quant-ph/0607167,
  title                     = {Quantitative entanglement witnesses},
  Author                   = {J. Eisert and F. G. S. L. Brandao and K. M.R. Audenaert},
  Journal                  = {New J. Phys.},
  doi={10.1088/1367-2630/9/3/046},
  Year                     = {2007},
  pages                  = {46},
  Volume                   = {9}
}

\end{document}